\documentclass[10pt,conference]{IEEEtran}

\usepackage{enumitem,verbatim}
\usepackage{amsmath,amsfonts,amssymb,amsthm}
\usepackage{proof,bussproofs}
\usepackage{algorithm}
\usepackage{algpseudocode}
\usepackage{algorithmicx, xcolor}
\usepackage{proof,bussproofs}
\usepackage{fancyvrb}
\usepackage{graphicx}
\usepackage{mathpartir}
\usepackage{hyperref}
\usepackage{listings}
\usepackage{stmaryrd}
\usepackage{dutchcal}
\usepackage{supertabular}
\usepackage{cite}
\usepackage[Symbol]{upgreek}
\usepackage[numbers,sort]{natbib}

\usepackage[T1]{fontenc}
\usepackage{inconsolata}

\newtheorem{theorem}{Theorem}
\newtheorem{lemma}{Lemma}

\theoremstyle{definition}
\newtheorem{definition}{Definition}
\newtheorem{example}{Example}
\newtheorem{rem}{Remark}

\usepackage{iftex}
\ifpdf
  \usepackage{underscore}         
  \usepackage[T1]{fontenc}        
\else
  \usepackage{breakurl}           
\fi

\makeatletter
\def\infrule#1#2#3{\@ifnextchar[{\@infrule{#1}{#2}{#3}}{\@infrule{#1}{#2}{#3}[*]}}%

\def\@infrule#1#2#3[#4]{
 \par\bigbreak
 \vtop{
  \noindent\hangindent2em\hangafter2\leavevmode\null
  \irule{#1:}\kern.5em\textbf{#2}\\[\smallskipamount]
  \ifx*#4 
   \null\quad$#3$
  \else 
   \setbox30=\hbox{\qquad$#3$\qquad\noindent#4}%
   \ifdim\wd30>\hsize
    \null\quad$#3$\par\kern-\parskip\smallskip\noindent\hangindent1em\hangafter0#4
   \else
    \null\quad$#3$\qquad\noindent#4
   \fi
  \fi
 }%
} \makeatother

\makeatletter
\newcommand*\Suppressnumber{%
	\lst@AddToHook{OnNewLine}{%
		\let\thelstnumber\relax%
		\advance\c@lstnumber-\@ne\relax%
	}%
} \makeatother

\makeatletter
\newcommand*\Reactivatenumber{%
	\lst@AddToHook{OnNewLine}{%
		\let\thelstnumber\origthelstnumber%
		\advance\c@lstnumber\@ne\relax}%
} \makeatother

\newcommand{\N}{\mathbb{N}}
\newcommand{\s}{\mathcal{l}}
\newcommand{\q}{\mathcal{r}}
\newcommand{\w}{\mathcal{u}}
\newcommand{\x}{\chi}

\newcommand{\unau}{\mathsf{VNAU}}
\newcommand{\unaur}{\unau_{\R}}

\newcommand{\permef}{{\cdot}}
\newcommand{\permid}{\mathit{Id}}
\newcommand{\dotcup}{\,\mathaccent\cdot\cup\,}
\newcommand{\Lra}{\Longrightarrow}
\newcommand{\spsys}{;\;\allowbreak}
\newcommand{\supp}{\mathsf{supp}}
\newcommand{\head}[1]{\mathsf{head}(#1)}
\newcommand{\substils}[1]{\sigma_L(#1)}
\newcommand{\substirs}[1]{\sigma_R(#1)}
\newcommand{\len}[1]{|#1|}
\newcommand{\R}{\mathcal{R}}
\newcommand{\Rlcs}{\R_{\textsf{lcs}}}
\newcommand{\hseq}{\mathcal{H}}
\newcommand{\alignsym}{\mathcal{h}}
\newcommand{\alignpos}[2]{[#1, #2]}
\newcommand{\swap}[2]{(#1\: #2)}
\newcommand{\np}[2]{\langle #1, \allowbreak #2 \rangle}
\newcommand{\dom}{\mathsf{Dom}}
\newcommand{\ran}{\mathsf{Ran}}
\newcommand{\atoms}{\mathsf{Atoms}}
\newcommand{\ds}{\mathsf{ds}}
\newcommand{\id}{id}
\newcommand{\arule}[1]{(\irule{#1})}
\newcommand{\irule}[1]{{\tt #1}}
\newcommand{\incmt}{{++}}
\newcommand{\decmt}{{--}}

\allowdisplaybreaks
\begin{document}

\title{Generalization of Variadic Structures with Binders: A Tool for Structural Code Comparison}

\author{
	\IEEEauthorblockN{{Alexander Baumgartner}}
	\IEEEauthorblockA{{\textit{Instituto de Ciencias de la Ingenier\'{i}a}} \\
		{\textit{Universidad de O'Higgins}}\\
		{Rancagua, Chile} \\
		{ORCID 0000-0002-4757-5907}}
	\and
	\IEEEauthorblockN{{Temur Kutsia}}
	\IEEEauthorblockA{{\textit{Research Institute for Symbolic Computation}} \\
		{\textit{Johannes Kepler University}}\\
		{Linz, Austria} \\
		{ORCID 0000-0003-4084-7380}}
}

\maketitle

\begin{abstract}
This paper introduces a novel anti-unification algorithm for the generalization of variadic structures with binders, designed as a flexible tool for structural code comparison. By combining nominal techniques for handling variable binding with support for variadic expressions (common in abstract syntax trees and programming languages), the approach addresses key challenges such as overemphasis on bound variable names and difficulty handling insertions or deletions in code fragments. The algorithm distinguishes between atoms and two kinds of variables (term and hedge variables) to compute best generalizations that maximally preserve structural similarities while abstracting systematic differences. It also provides detailed information to reconstruct original expressions and quantify structural differences. This information can be useful in tasks like code clone detection, refactoring, and program analysis. By introducing a parametrizable rigidity function, the technique offers fine-grained control over similarity criteria and reduces nondeterminism, enabling flexible adaptation to practical scenarios where trivial similarities should be discounted. Although demonstrated primarily in the context of code similarity detection, this framework is broadly applicable wherever precise comparison of variadic and binder-rich representations is required.\end{abstract}

\begin{IEEEkeywords}
Variadic nominal terms, anti-unification, generalization, code clone detection,
program analysis
\end{IEEEkeywords}

\section{Introduction}

A generalization of two formal expressions $t_1$ and $t_2$ is an expression $r$ that maximally preserves similarities between $t_1$ and $t_2$, while abstracting over their differences in a systematic way. Anti-unification \cite{Plotkin70,Reynolds70,DBLP:conf/ijcai/CernaK23} is a logic-based technique used to compute generalizations. Usually, its output includes not only the generalized expression $r$, but also a description of how $t_1$ and $t_2$ differ, along with mappings that reconstruct each original expression from the generalization. The notions of similarity and difference can be parameterized and quantified, which makes anti-unification applicable in various areas of computer science, artificial intelligence, computational linguistics, etc. 

One notable application area is software science, where anti-unification techniques are widely used in various automated programming and program analysis tasks. Recent examples include bottom-up synthesis of programs with local variables~\cite{DBLP:journals/pacmpl/LiZDZW24}, dynamic analysis of machine learning programs~\cite{DBLP:conf/icse/ZhengS24}, synthesis of modified-privilege access control policies~\cite{DBLP:journals/pacmpl/DAntoniDGRRS24}, learning-based software vulnerability analysis~\cite{DBLP:conf/icse/NongFYZL0C24}, and metaprogram learning~\cite{DBLP:journals/nature/RulePCEKT24}. In program analysis, anti-unification has been successfully integrated into tools for automated program repair such as Getafix~\cite{DBLP:journals/pacmpl/BaderSPC19}, Rex~\cite{DBLP:conf/nsdi/MehtaBKMBABABK20}, Fixie~\cite{DBLP:conf/sigsoft/WinterNBCHHWKWM22}, and Revisar~\cite{DBLP:conf/sbes/SousaSGBD21}, as well as for software code clone detection and refactoring~\cite{DBLP:conf/iwsc/YernauxV22,DBLP:conf/ershov/BulychevKZ09,DBLP:journals/programming/ThompsonLS17}.

Within these applications, anti-unification techniques are typically limited to relatively simple variants, which have difficulties in identifying similarities between code fragments that contain non-trivial differences, such as insertions or deletions. Furthermore, when assessing similarities between code pieces, differences in the names of local (bound) variables are often overemphasized, receiving greater weight than is warranted.

To address these and related challenges, we propose a new anti-unification technique that simultaneously handles binders and variadic expressions, combining the strengths of both approaches. Concepts such as variable binding, scope, renaming, capture-avoiding substitution, and fresh symbol generation are fundamental in code analysis and transformation. To work with these reliably, we adopt a mathematically rigorous foundation based on nominal techniques \cite{gabbay,DBLP:conf/lics/GabbayP99}. Variadic expressions, which allow a variable number of arguments, are playing an increasingly significant role in software science across areas such as language design, formal reasoning, metaprogramming, and practical software engineering. For example, they enable flexible APIs, they are used in formats like XML/HTML and abstract syntax trees with arbitrary numbers of child nodes, and they are helpful in macros and templates for writing reusable and composable code structures.
Moreover, program analysis tools, including clone detectors, refactoring engines, and symbolic execution frameworks, can greatly benefit from variadic representations, as they must analyze and transform code structures with arbitrary arity in real-world software.

Formal symbolic techniques, in particular generalization, have been studied for nominal and variadic expressions, but mostly in isolation \cite{DBLP:conf/fscd/Schmidt-Schauss22,DBLP:conf/unif/BaumgartnerN20,DBLP:conf/rta/BaumgartnerKLV15,DBLP:journals/jar/KutsiaLV14,DBLP:journals/iandc/BaumgartnerK17,DBLP:conf/jelia/BaumgartnerK14,RISC5180,DBLP:conf/ilp/YamamotoIIA01}. Their combination, however, remains largely unexplored, with the exception of \cite{DBLP:conf/tbillc/DunduaKR19}, which addresses the unification problem in this setting.

The algorithm presented in this paper operates on \emph{variadic terms} (trees) that may include \emph{binders}. It makes a syntactic distinction between \emph{atoms}, which can be bound, and \emph{variables}, which are subject to substitution.
\footnote{In the context of program analysis, atoms may represent program entities such as local variables or function arguments, while variables serve as meta-variables used to abstract over code fragments.}
That is, a combination of variadic and nominal terms, called \emph{variadic nominal terms}, is being considered.
To fully exploit variadicity, we distinguish between two kinds of variables: \emph{individual variables}, which can be substituted by an individual term, and \emph{hedge variables}, which can be substituted by a finite, possibly empty, sequence of terms, referred to as a \emph{hedge}.
The algorithm computes least general generalizations for such variadic nominal terms: these generalizations are more general than the input expressions, and no strictly more specific generalization exists.

In addition to computing generalizations, the algorithm also captures information about how the original expressions differ and how they can be reconstructed from the generalization. This difference information is particularly useful for applications that aim to quantify the similarity or dissimilarity between structures.

To provide finer control over what is considered meaningful similarity during generalization (and to reduce nondeterminism), the algorithm is parametrized by a special function called the rigidity function. At each level of the term-trees being generalized, this function determines which nodes should be retained in the generalization. By varying the choice of rigidity function, we can obtain generalizations of differing granularity or precision. This flexibility is particularly useful in applications such as code clone detection, where it may be desirable to exclude trivial or insignificant similarities from being classified as clones.

Although our illustrative examples focus on code clone detection, we emphasize that the applicability of the technique developed in this paper extends well beyond this domain. The method provides a flexible and general-purpose framework for characterizing similarities and differences between formal structures with binders and variadic arity. As such, it can be applied in a wide range of contexts, including program analysis, automated refactoring, symbolic reasoning, and even in linguistic or mathematical domains where structural comparison and abstraction over variable binding and variadic constructs are essential.

The paper is organized as follows. Section~\ref{sect:example} presents a running example related to software code clone detection. Section~\ref{sec:preliminaries} introduces the necessary notation and terminology, and also demonstrates how the example from Section~\ref{sect:example} can be modeled using variadic nominal syntax. In Section~\ref{sec:au-general}, we present a general algorithm for variadic nominal anti-unification, prove its termination, soundness, and completeness, and discuss its strengths and limitations. Section~\ref{sec:au-rigid} refines this algorithm into a more efficient and flexible version by introducing a rigidity function as a parameter. An additional extension of the algorithm is presented in Section~\ref{sec:au-x}. Section~\ref{sect:future} outlines a possible direction for future work, and Section~\ref{sec:conc} concludes the paper.

\section{Running Example: Detecting Code clones}
\label{sect:example}

Software code clones are similar pieces of code, usually obtained by reusing: copying and pasting with or without minor adaptation. They have been studied from different perspectives, e.g., in
\cite{DBLP:journals/scp/RoyCK09,DBLP:journals/spe/Yang91,DBLP:journals/sqj/EvansFM09,DBLP:conf/icsm/BaxterYMSB98,DBLP:conf/wcre/KoschkeFF06,DBLP:conf/ershov/BulychevKZ09,DBLP:conf/padl/LiT10}.
Clones are classified into the following categories:
\begin{description}
	\item[Type-1:] Identical code fragments except for variations in whitespace, layout and comments.
	\item[Type-2:] Syntactically identical fragments except for variations in identifiers, literals, types, whitespace, layout and
	comments.
	\item[Type-3:] Copied fragments with further modifications such as changed, added or removed statements, in addition to
	variations of type-2.
	\item[Type-4:] Two or more code fragments that perform the same computation but are implemented by different syntactic
	variants.
\end{description}
An illustrative comparison of the clone types is given in \cite{DBLP:journals/scp/RoyCK09}. The example shown in Figure~\ref{fig:sw-clone-intro} is taken from there with a minor modification. It shows a piece of code and two of its clones, one of type-2 and one of type-3. Below we will show how this example can be encoded in out syntax and illustrate the use of anti-unifucation to analyze similarity of the provided fragments.

\begin{figure}[ht!]
\setlength{\tabcolsep}{0pt}
\centering
\begin{tabular}{c}
\begin{lstlisting}
void sumProd(int n) {
	float sum = 0.0;
	float prod = 1.0;
	for(int i=1; i<=n; i++) {
		sum = sum + i;
		prod = prod * i;
		foo(sum, prod); }}

void sumProd(int |\bfseries a|) {
	|\bfseries double s| = 0.0;
	|\bfseries double p| = 1.0;
	for(int |\bfseries j\,|=1; |\bfseries j\,|<=|\bfseries a\,|; |\bfseries j\,|++) {
		|\bfseries s\,| = |\bfseries s\,| + |\bfseries j\,|;
		|\bfseries p\,| = |\bfseries p\,| * |\bfseries j\,|;
		foo(|\bfseries s\,|, |\bfseries p\,|); }}

 void sumProd(int |\bfseries a|) {
	 |\bfseries double s| = 0.0;
	 |\bfseries double p| = 1.0;
	 for(int |\bfseries j\,|=1; |\bfseries j\,|<=|\bfseries a\,|; |\bfseries j\,|++) {
		 |\bfseries s\,| = |\bfseries s\,| + |\bfseries j\,|;
		 |\bfseries// line deleted|
		 foo(|\bfseries j|, |\bfseries s\,|, |\bfseries p\,|, |\bfseries a|); }}
\end{lstlisting}
\end{tabular}

\caption{Original code, type-2 clone, and type-3 clone.}
\label{fig:sw-clone-intro}
\end{figure}

\section{Notions, Notation, Terminology}\label{sec:preliminaries}

This section defines important notions that are used throughout the work. For more details we refer the reader to related work, e.g., \cite{DBLP:journals/tocl/LevyV12,DBLP:journals/jar/KutsiaLV14,DBLP:conf/rta/BaumgartnerKLV15,DBLP:conf/tbillc/DunduaKR19}. 

The signature of variadic nominal terms considers four pairwise disjoint sets:\footnote{For simplicity, we consider here the unsorted case.} a countable infinite set of \emph{atoms} ${\pmb A}=\{a,b,c,\ldots\}$, a countable set of \emph{variadic function symbols} ${\pmb \Sigma}=\{f,g,h,\ldots\}$, a countable infinite set of  \emph{individual variables} ${\pmb x}=\{x,y,z,\ldots\}$, and a countable infinite set of  \emph{hedge variables} ${\pmb X}=\{X,Y,Z,\ldots\}$. 
Like in the case of (ranked) nominal terms, we consider \emph{permutations} of atoms, say  $\pi$, with finite support $\supp(\pi) := \{a \in \pmb{A} \mid  \pi(a) \neq  a\}$.
The inverse of $\pi$ is denoted by $\pi^{-1}$, and {$\permid$} is the identity permutation (i.e., $\supp(\permid) = \emptyset$).
Permutations are represented by a finite sequence of swappings, e.g., $\swap{a}{b}\swap{a}{c}$. By doing so, permutation composition simply reduces to sequence concatenation, and the inverse corresponds to reversing the sequence.

\emph{Variadic nominal terms} (terms, for short) are either individual terms $t$ or hedges (finite sequences) $\tilde s$, given by the grammar:
\begin{align*}
t::= & \; a \mid \ \pi \permef x \mid \ a.t \mid  \ f(\tilde s),
\qquad
r::= \pi \permef X \mid \ t,
\\
\tilde s::= & \; r_1, \dots, r_n, \quad n\geq 0,
\end{align*}
where $a$ is an \emph{atom}, $\pi\permef x$ is an \emph{individual suspension}, $a.t$ denotes the \emph{abstraction} of atom $a$ in the individual term~$t$, $f(\tilde s)$~is a \emph{function application} of the variadic function $f$ to a (possibly empty) hedge of arguments~$\tilde s$, and $\pi \permef X$ is a \emph{hedge suspension}. A singleton hedge that consists of an individual term $t$ is considered the same as $t$ itself. Hedges are assumed to be flat, i.e., there is no difference between $\tilde{s_1},(\tilde{s_2}),\tilde{s_3}$ and $\tilde{s_1},\tilde{s_2},\tilde{s_3}$. The empty hedge is denoted by $\varepsilon$. We use $\s,\q,\w$ to emphasize that we are talking about either type of term, i.e., a hedge or an individual term. The notion \emph{variable} refers to either a hedge or an individual variable, and $\x$ denotes some variable from $\pmb x\cup\pmb{X}$. \emph{Suspension} simply refers to either an individual or a hedge suspension.

\begin{example}
    Figure \ref{fig:sw-clone-hedges} illustrates the representation of code piece from Figure~\ref{fig:sw-clone-intro} as an abstract syntax tree in the language of variadic nominal terms. In particular, it shows the translation of the original code and the type-3 clone. Variable names become atoms that are bounded by the enclosing block statement.

\begin{figure}[ht!]\centering
\begin{tabular}{l}
	$n.sum.prod.$sumProd(input(type(int), $n$), returnType(void),\\
	\qquad  {=}(type(float), $sum$, 0.0),\\
	\qquad  {=}(type(float), $prod$, 1.0),\\
	\qquad  $i.$for({=}(type(int),$i$,1), {$\leq$}($i,n$), \incmt($i$),\\
	\qquad\qquad {=}($sum$, +($sum,i$)),\\
	\qquad\qquad {=}($prod$, *($prod,i$)),\\
	\qquad\qquad foo($sum,prod$)))
\end{tabular}

\medskip
\begin{tabular}{l}
	$a.s.p.$sumProd(input(type(int), $a$), returnType(void),\\
	\qquad  {=}(type(double), $s$, 0.0),\\
	\qquad  {=}(type(double), $p$, 1.0),\\
	\qquad  $j.$for({=}(type(int),$j$,1), {$\leq$}($j,a$), \incmt($j$),\\
	\qquad\qquad {=}($s$, +($s,j$)),\\
    \qquad\qquad {=}($p$, *($p,i$)),\\
	\qquad\qquad foo($s,p$)))
\end{tabular}

\medskip
\begin{tabular}{l}
	$a.s.p.$sumProd(input(type(int), $a$), returnType(void),\\
	\qquad  {=}(type(double), $s$, 0.0),\\
	\qquad  {=}(type(double), $p$, 1.0),\\
	\qquad  $j.$for({=}(type(int),$j$,1), {$\leq$}($j,a$), \incmt($j$),\\
	\qquad\qquad {=}($s$, +($s,j$)),\\
	\qquad\qquad foo($j,s,p,a$)))
\end{tabular}
\caption{The original code and its type-2 and type-3 clones as variadic nominal terms.
\label{fig:sw-clone-hedges}
}
\end{figure}

\end{example}

The \emph{effect of a swapping} over an atom is defined by $\swap{a}{b} \permef a = b$, $\swap{a}{b}\permef b = a$ and $\swap{a}{b}\permef c = c$, when $c \notin \{a, b\}$. It is extended to the rest of terms: $\swap{a}{b}\permef (\pi \permef \x) = \swap{a}{b}\pi \permef \x$ (where $\swap{a}{b}\pi$ is the composition of $\pi$ and $\swap{a}{b}$), $\swap{a}{b}\permef (c.t) = \left ( \swap{a}{b}\permef c \right).\left (\swap{a}{b} \permef t \right )$, $\swap{a}{b}\permef f(\tilde s) = f(\swap{a}{b}\permef \tilde s)$, and $\swap{a}{b}\permef (r_1,\ldots,r_n) = (\swap{a}{b}\permef r_1,\ldots, \swap{a}{b}\permef r_n)$. 
The \emph{effect of a permutation} is defined by $\swap{a_1}{b_1} \cdots \swap{a_n}{b_n}\permef t = \swap{a_1}{b_1}\permef \left (\swap{a_2}{b_2} \cdots \swap{a_n}{b_n}\permef t\right )$. The effect of the empty permutation is $\permid \permef \s = \s$. We extend it to suspensions and write $\x$ as the shortcut of $\permid\permef \x$.

The $i$th element of a hedge $\tilde s$ is denoted by $\tilde s\vert_i$, and $\tilde s\vert_i^j$ denotes the subhedge (i.e., substring hedge) from the $i$th element to the $j$th element, both included. The case $j<i$ corresponds to the empty hedge $\varepsilon$. The \emph{length} of a hedge $\tilde s$ is the number of elements in it, written as $\len{\tilde s}$. Individual terms and hedge suspensions are treated as singleton hedges, i.e., the above notions naturally extend to them.
The \emph{head} of an individual term or a hedge suspension $r$, denoted by $\head{r}$, is defined as $\head{a}=a$, $\head{\pi\permef\nobreak \x}=\x$, $\head{a.t} = .$, and $\head{f(\tilde s)}=f$.\footnote{It is assumed that some special characters, e.g., "\textvisiblespace.,()$\cdot\bullet$", do not occur in the signature.}

{\em Substitutions}, denoted by $\sigma, \gamma, \ldots$, are finite mappings from variables to terms, under the condition that individual variables are mapped to individual terms. They are defined by a finite set of \emph{assignments} of the form $\{\x_1\mapsto\w_1,\dots \x_n\mapsto\w_n\}$.
Given a substitution $\sigma=\{\x_1\mapsto\w_1,\dots \x_n\mapsto\w_n\}$, the \emph{domain} of $\sigma$ is the finite set $\dom(\sigma):=\{\x_1,\dots,\x_n\}$, and the \emph{range} is $\ran(\sigma):=\{\w_1,\dots,\w_n\}$.
For all variables $\x\in ({\pmb x}\cup{\pmb X})\setminus\dom(\sigma)$, the mapping is extended as being the identity suspension $\permid\permef\x$.
The variables in $\dom(\sigma)$ are also said to be {\em instantiated} by $\sigma$.
The identity substitution with the empty domain is denoted by $\id$.
The {\em effect of a substitution} $\sigma$ on a term $\w$, denoted as~$\w\sigma$, is inductively defined by
\begin{gather*}
a\sigma = a,\qquad
(\pi \permef \x)\sigma = \pi \permef \x\sigma,\qquad
(a.t)\sigma = a.(t\sigma),
\\
f(\tilde{s})\sigma = \tilde{s}\sigma,\qquad
(r_1,\ldots, r_n)\sigma = r_1\sigma,\ldots, r_n\sigma.
\end{gather*}

\emph{Freshness constraints} are pairs of the form $a\#\x$ ($a$ is fresh for $\x$), forbidding the free occurrence of some atom~$a$ in instantiations of some variable~$\x$. A \emph{freshness context} (denoted by $\nabla, \Gamma,\dots$) is a finite set of freshness constraints. Given a substitution $\sigma$ and a freshness context $\nabla=\{a_1\#\x_1,\dots,a_n\#\x_n\}$, we say that $\sigma$ respects $\nabla$
iff $\nabla\vdash a_i\#\x_i\sigma$, for all $1\leq i\leq n$, where $\nabla\vdash a\# \s$ is defined as follows:
{
\begin{prooftree}
\AxiomC{$a\neq b$}
\RightLabel{\arule{\#atom}}
\UnaryInfC{$\nabla \vdash a\# b$}
\DisplayProof
\hspace{0.5cm}
\AxiomC{ $\nabla\vdash a\# \tilde s$}
\RightLabel{\arule{\#appl}}
\UnaryInfC{ $ \nabla \vdash f(\tilde s)$}
\end{prooftree}
\begin{prooftree}
\AxiomC{$ \nabla \vdash a\#r_1 \quad \cdots \quad \nabla \vdash a\#r_n$}
\RightLabel{\arule{\#hedge}}
\UnaryInfC{$ \nabla \vdash a\#(r_1,\ldots, r_n)$}
\end{prooftree}
\begin{prooftree}
\AxiomC{\phantom{X}}
\RightLabel{\arule{\#abs1}}
\UnaryInfC{$ \nabla \vdash a\# a.t$}
\DisplayProof
\hspace{0.5cm}
\AxiomC{ $a\neq b \qquad \nabla\vdash a\# t$}
\RightLabel{\arule{\#abs2}}
\UnaryInfC{ $ \nabla \vdash a\#b.t$}
\end{prooftree}
\begin{prooftree}
\AxiomC{$ (\pi^{-1}\permef a\# \x)\in \nabla$}
\RightLabel{\arule{\#susp}}
\UnaryInfC{ $ \nabla \vdash a\#\pi\permef \x$}
\end{prooftree}
}
Given a freshness context $\nabla$ and a substitution $\sigma$ that respects $\nabla$, the {\em instance of $\nabla$ under $\sigma$}, written~$\nabla\sigma$, is the smallest freshness context $\Gamma$ such that $\Gamma\vdash a\#\x\sigma$ for all $a\# \x \in \nabla$. Note that $\Gamma$ can easily be obtained by the above rules.

The \emph{difference set} of two permutations $\pi$ and $\pi'$ is defined as $\ds(\pi, \pi') = \{a\in {\pmb A} \mid \pi\permef a \neq \pi'\permef a\}$.
Intuitively, the \emph{equivalence} relation $\approx$, defined by the rules below, states that two terms or hedges are considered equal if they only differ in bound atom names. Permutations are used to rename atoms.
{
\begin{mathpar}
\inferrule*[RIGHT=\arule{\approx\!atom}]{\quad }{ \nabla\vdash a\approx a}  %
\and
\inferrule*[RIGHT=\arule{\approx\!appl}]{ \nabla \vdash \tilde s \approx \tilde q}{ \nabla \vdash f(\tilde s)\approx f(\tilde q)}
\and
\inferrule*[RIGHT=\arule{\approx\!hedge}]{ \nabla \vdash r_1\approx r'_1 \quad \cdots \quad \nabla \vdash r_n\approx r'_n}{ \nabla \vdash r_1,\ldots r_n\approx r'_1,\ldots,r'_n}
\and
\inferrule*[RIGHT=\arule{\approx\!abs1}]{\nabla \vdash t\approx t'}{\nabla \vdash a.t \approx a.t'}
\and
\inferrule*[RIGHT=\arule{\approx\!abs2}]{\nabla \vdash t\approx \swap{a}{b}\permef t' \qquad \nabla \vdash a\# b.t'}{\nabla \vdash a.t \approx b.t'}
\and
\inferrule*[RIGHT=\arule{\approx\!susp}]{a\#\x \in \nabla \text{ for all }a\in\ds(\pi,\pi')}{\nabla \vdash \pi\permef \x \approx \pi' \permef \x} %
\end{mathpar}
}

A \emph{term-in-context} is a pair of a freshness context $\nabla$ and a term $\w$, written $\np{\nabla}{\w}$. We use $T,U,G$ to denote terms-in-context.
\begin{definition}
\label{def:more-general-relation}
A term-in-context $\np{\nabla_1}{\w_1}$ is {\em more general} than a term-in-context $\np{\nabla_2}{\w_2}$, written $\np{\nabla_1}{\w_1}\preceq\np{\nabla_2}{\w_2}$, if there exists a substitution $\sigma$ that respects $\nabla_1$, such that $\nabla_1\sigma\subseteq \nabla_2$ and $\nabla_2\vdash \w_1\sigma\approx \w_2$.
$\np{\nabla_1}{\w_1}$ and $\np{\nabla_2}{\w_2}$ are {\em equi-general}, written $\np{\nabla_1}{\w_1}\simeq \np{\nabla_2}{\w_2}$ iff $\np{\nabla_1}{\w_1}\preceq\np{\nabla_2}{\w_2}$ and $\np{\nabla_2}{\w_2}\preceq\np{\nabla_1}{\w_1}$. As usual, $\prec$ denotes $\preceq\setminus\simeq$.
\end{definition}

\begin{definition}
A term-in-context $G$ is a \emph{generalization} of the terms-in-context $T$ and $U$ if it is more general than both of them, i.e.,  $G\preceq T$ and $G\preceq U$. $G$ is a {\em least general generalization} ($lgg$) of $T$ and $U$ if there exists no other generalization $G'$ of $T$ and $U$ such that $G\prec G'$.
\end{definition}

For instance, consider the terms-in-context $T=\np{\emptyset}{f(a,\allowbreak b,a,b)}$ and $U=\np{\emptyset}{f(c,c)}$. The term-in-context $\np{\nabla}{t}$ where $\nabla=\{b\# x, a\#X, c\# X\}$ and $t=f(x,X,x,X)$ is more general than $T$, since $\sigma=\{x\mapsto a, X\mapsto b\}$ respects $\nabla$, $\nabla\sigma\subseteq\emptyset$, and $t\sigma = f(a,b,a,b)$. It is also more general than $U$, since $\sigma'=\{x\mapsto c, X\mapsto \varepsilon\}$ respects $\nabla$, $\nabla\sigma'\subseteq\emptyset$, and $t\sigma' = f(c,c)$. Therefore, $\np{\nabla}{t}$ is a generalization of $T$ and $U$.
Note that, $\np{\nabla}{f(X,X)}$ and $\np{\{d\# X\}\cup\nabla}{t}$ are also generalizations of $T$ and $U$. Moreover, $\np{\nabla}{f(X,X)}\prec\np{\nabla}{t}\prec\np{\{d\# X\}\cup\nabla}{t}$.
Observe that, by simply adding a freshness constraint, a strictly less general term-in-context can be obtained. This is a crucial observation that has been discussed in~\cite{DBLP:conf/rta/BaumgartnerKLV15}. It leads to $\preceq$ not being well-founded, i.e., no lgg exists.
Therefore, it has been suggested to consider a \emph{finite set} of atoms $A\subset \pmb{A}$ in the generalizations.

The set of atoms that occur in a suspension $\pi\permef\x$ is defined as
$\supp(\pi)$. 
We use $\atoms(\w)$, $\atoms(\nabla)$, $\atoms(\np{\nabla}{\w})$ and $\atoms(\sigma)$ to denote the set of all atoms that occur in a term $\w$,
a freshness context $\nabla$, a term-in-context $\np{\nabla}{\w}$ and in $\ran(\sigma)$, respectively.
A term $\w$, a freshness context $\nabla$, a term-in-context $\np{\nabla}{\w}$ or a substitution $\sigma$ is \emph{based on a set of atoms~$A$} ($A$-based), respectively, iff $\atoms(\w)\subseteq A$, $\atoms(\nabla)\subseteq A$, $\atoms(\np{\nabla}{\w})\subseteq A$ or $\atoms(\sigma)\subseteq\nobreak A$.
A~permutation $\pi$ is $A$-based iff $\supp(\pi)\subseteq A$.
We~consider the anti-unification problem that is parametric in a finite set of atoms $A\subset \pmb{A}$ and everything is $A$-based:\vspace{2mm}
\begin{description}
	\item[Given:]
	Two $A$-based terms-in-context $\np{\nabla}{\w_1}$ and $\np{\nabla}{\w_2}$.
	\item[Find:]
	An $A$-based lgg of $\np{\nabla}{\w_1}$ and $\np{\nabla}{\w_2}$.
\end{description}
\vspace{2mm}
Intuitively, the choice of $A$ should be such that all the structural similarities between $\w_1$ and $\w_2$ can be captured by an $A$-based generalization.
This is guaranteed by selecting $A=A_1\cup A_2$ where $A_1=\atoms(\np{\nabla}{\w_1})\cup\atoms(\np{\nabla}{\w_2})$ and $A_2\subset\pmb{A}\setminus A_1$ is a set of $k$ atoms where $k$ is the minimum amount of abstraction occurrences in $\w_1$ and in $\w_2$. The proof and more details can be found in~\cite{DBLP:conf/rta/BaumgartnerKLV15}.

\section{Variadic Nominal Anti-Unification ($\unau$)}\label{sec:au-general}
The rule-based algorithm $\unau$, described in this section, combines techniques from \cite{DBLP:journals/jar/KutsiaLV14} and \cite{DBLP:conf/rta/BaumgartnerKLV15}.

The rules operate on states of the form $P \spsys S \spsys \Gamma \spsys \sigma$, where $P$ is the problem set, $S$ is the store of already solved problems, $\Gamma$ represents the freshness context computed so far, and $\sigma$ is a substitution that represents the generalization computed so far. $P$ and $S$ are sets of anti-unification problems (AUPs) of the form $\x\colon \s \triangleq \q$, where $\x$ is the generalization variable and $\s,\q$ are hedges or individual terms to be generalized. The notion of being based on a finite set of atoms $A$ (i.e., $A$-based) is naturally extended to AUPs and states, but we won't mention it, if it's clear from the context or irrelevant.

Given two $A$-based terms $\s$, $\q$, and an $A$-based freshness context $\nabla$, we create the \emph{initial state} 
$\{X\colon \s \triangleq \q\}\spsys \emptyset \spsys \emptyset \spsys \id$, where $X$ is a hedge variable, that neither occurs in $\s, \q$, nor in $\nabla$. Then the state transformation rules defined below are being applied exhaustively, in all possible ways. When no more rule is applicable to a state, then a \emph{final state} has been reached. $A$ and $\nabla$ are global parameters and unaffected by rule applications.
A variable is said to be \emph{new} if it neither occurs in the current state nor in a global parameter. $\dotcup$ is the disjoint union operation.

\infrule{{Tri-T}}{Trivial Term Elimination}
{
\{
X\colon a \triangleq a
\} \dotcup P \spsys S \spsys \Gamma \spsys \sigma
 \Lra
P \spsys S \spsys \Gamma \spsys \sigma\{X\mapsto a\}
.}

\infrule{{Tri-H}}{Trivial Hedge Elimination}
{
\{
X\colon \varepsilon \triangleq \varepsilon
\} \dotcup P \spsys S \spsys \Gamma \spsys \sigma
 \Lra
P \spsys S \spsys \Gamma \spsys \sigma\{X\mapsto \varepsilon\}
.}

\infrule{{Dec-T}}{Term Decomposition}
{
\{
X\colon f(\tilde s) \triangleq f(\tilde q)
\} \dotcup P \spsys S \spsys \Gamma \spsys \sigma
 \Lra\\
\{
Y\colon \tilde s \triangleq \tilde q
\} \cup P \spsys S \spsys \Gamma \spsys \sigma\{X\mapsto f(Y)\}
,}
[where $Y$ is a new hedge variable.]

\infrule{{Dec-H}}{Hedge Decomposition}
{
\{
X\colon \tilde s_1,\tilde s_2 \triangleq \tilde q_1,\tilde q_2
\} \dotcup P \spsys S \spsys \Gamma \spsys \sigma
 \Lra\\
\{
Y_1\colon \tilde s_1 \triangleq \tilde q_1,\:\allowbreak
Y_2\colon \tilde s_2 \triangleq \tilde q_2
\} \cup P \spsys S \spsys \Gamma \spsys \sigma\{X\mapsto Y_1,Y_2\}
,}
[where $\len{\tilde s_1, \tilde q_1}=1$ or $\len{\tilde s_1}=\len{\tilde q_1}=1$. Moreover, $\len{\tilde s_2, \tilde q_2}\geq1$ and $Y_1,Y_2$ are new hedge variables.]

\infrule{{Abs-T}}{Term Abstraction}
{
	\{
	X\colon  a.t_1 \triangleq b.t_2
	\}\dotcup P \spsys S \spsys \Gamma \spsys \sigma
	\Lra\\
	\{
	Y\colon \swap{c}{a}\permef t_1 \triangleq \swap{c}{b}\permef t_2
	\}\cup P \spsys S \spsys \Gamma \spsys \sigma\{X\mapsto c.Y\},
}
[where $c\in A$, $Y$ is a new hedge variable, $\nabla \vdash c\# a.t_1$  and $\nabla \vdash c\# b.t_2$.]

\infrule{{Sol-T}}{Term Solving}
{
\{
X\colon r_1 \triangleq r_2
\} \dotcup P \spsys S \spsys \Gamma \spsys \sigma
 \Lra
P \spsys 
\{
X\colon r_1 \triangleq r_2
\} \cup S \spsys \{ a\#X \mid a\in A, \nabla \vdash a\#r_1 \,\wedge\, \nabla \vdash a\#r_2\} \cup \Gamma \spsys \sigma
,}
[where $r_1$ and $r_2$ are suspensions or $\head{r_1}\neq\head{r_2}$.]

\infrule{{Sol-H}}{Hedge Solving}
{
\{
X\colon \tilde s \triangleq \tilde q
\} \dotcup P \spsys S \spsys \Gamma \spsys \sigma
 \Lra
P \spsys 
\{
X\colon \tilde s \triangleq \tilde q
\} \cup S \spsys \{ a\#X \mid a\in A, \nabla \vdash a\#\tilde s \,\wedge\, \nabla \vdash a\#\tilde q\} \cup \Gamma \spsys \sigma
,}
[where $\len{\tilde s, \tilde q}=1$.]

\infrule{{Mer}}{Term or Hedge Merging}
{
P \spsys
\{
\x_1\colon \s_1 \triangleq \q_1,
\x_2\colon \s_2 \triangleq \q_2
\} \dotcup S \spsys \Gamma \spsys \sigma 
 \Lra\\
P \spsys \{{\x_1 : \s_1 \triangleq \q_1}\} \cup S \spsys
\Gamma\{\x_2 \mapsto\pi\permef \x_1\} \spsys \sigma\{\x_2\mapsto \pi\permef \x_1\}
,}
[where $\pi$ is an $A$-based permutation such that $\nabla\vdash \pi\permef \s_1\approx \s_2$ and $\nabla\vdash \pi\permef \q_1\approx \q_2$.]

\infrule{{Nar-T}}{Term Narrowing}
{
P \spsys
\{
X\colon t_1 \triangleq t_2
\} \dotcup S \spsys \Gamma \spsys \sigma
 \Lra\\
P \spsys \{
x\colon t_1 \triangleq t_2
\}\cup S \spsys \Gamma\{X\mapsto x\} \spsys \sigma\{X\mapsto x\}
.}

\medskip
When the transformation process of the initial state $\{X\colon \s \triangleq\nobreak \q\};\, \emptyset ;\, \emptyset ;\, \id$ of some $A$-based $\s$, $\q$ and $\nabla$ leads to a final state $P ;\, S ;\, \Gamma ;\, \sigma$, then the computed $A$-based generalization is $\np{\Gamma}{X\sigma}$. Moreover, $S$ contains all the differences of the input terms $\s$ and $\q$. Note that, for every final state $P=\emptyset$. The rule \irule{Mer} involves deciding if an $A$-based permutation $\pi$ exists such that $\nabla\vdash \pi\permef \s_1\approx \s_2$ and $\nabla\vdash \pi\permef \q_1\approx \q_2$. This decision problem is called the equivariance problem. In~\cite{DBLP:conf/rta/BaumgartnerKLV15} an algorithm is given for solving the equivariance problem. It can be trivially adapted to our setting, since there are no variable instantiations involved.

Since decomposing the hedges involves branching, all the produced generalizations are collected. As usual, some of the computed generalizations may not be lggs. However, below 
we will prove that the computed set of generalizations is complete.
Minimizing the complete set of generalization can be done subsequently, utilizing a matching algorithm.

Now we illustrate the $\unau$ algorithm with an example. For the sake of readability, we omit irrelevant assignments in the computed substitution. In particular, only the assignment that corresponds to the generalization variable $X$ which is used in the initial state will be illustrated.

 \begin{example}\label{ex:simple-two-lggs}
 Let $s=c.f(a,c)$ and $t=b.f(b,c)$ with $A=\{a,b,c\}$ and $\nabla=\emptyset$. The algorithm generates eight branches 
 from which the following two lead to $A$-based lggs:
\begin{alignat*}{4}
\text{Branch 1:} 
    & \hspace{4pt} & & \{X \colon c.f(a,c) \triangleq b.f(b,c) \} \spsys \emptyset \spsys \emptyset \spsys \id \\
\;\Lra_{\text{\sf Abs-T}}
    & \hspace{4pt} & & \{Y \colon f(a,b) \triangleq f(b,c) \} \spsys \emptyset \spsys \emptyset \spsys \{X\mapsto b.Y \} \\
\;\Lra_{\text{\sf Dec-T}}
    & \hspace{4pt} & & \{Z \colon (a,b) \triangleq (b,c) \} \spsys \emptyset \spsys \emptyset \spsys \{X \mapsto b.f(Z), \ldots\}  \\
\;\Lra_{\text{\sf Dec-H}}
    & \hspace{4pt} & & \{V \colon a \triangleq \varepsilon, \, V' \colon b \triangleq (b,c) \} \spsys \emptyset \spsys \emptyset \spsys 
    \\ & \hspace{4pt} & & \{ X\mapsto b.f(V, V'), \ldots\} \\
\;\Lra_{\text{\sf Sol-H}}
    & \hspace{4pt} & & \{ V' \colon b \triangleq (b,c) \} \spsys \{V \colon a \triangleq \varepsilon\} \spsys \{b\#V, c\#V\} \spsys 
    \\ & \hspace{4pt} & & \{ X\mapsto b.f(V,V'), \ldots\} \\
\;\Lra_{\text{\sf{Dec-H}}}
    & \hspace{4pt} & & \{ W'\!\colon b \triangleq b,  W\!\colon \varepsilon \triangleq c \} ; \{V\!\colon a \triangleq \varepsilon\} ; \{b\#V, c\#V\} ;
    \\ & \hspace{4pt} & & \{ X\mapsto b.f(V,W',W), \ldots\} \\
\;\Lra_{\text{\sf{Tri-T}}}
    & \hspace{4pt} & & \{ W \colon \varepsilon \triangleq c \} \spsys \{V \colon a \triangleq \varepsilon\} \spsys \{b\#V, c\#V \} \spsys
    \\ & \hspace{4pt} & & \{ X\mapsto b.f(V,b,W), \ldots\} \\
\;\Lra_{\text{\sf{Sol-H}}}
    & \hspace{4pt} & & \emptyset \spsys \{V \colon a \triangleq \varepsilon, \,W \colon \varepsilon \triangleq c\} \spsys \{b\#V, c\#V, 
    \\ & \hspace{4pt} & & a\#W, b\#W\} \spsys \{ X\mapsto b.f(V,b,W), \ldots\}.\\[1mm]
\text{Branch 2:} 
    & \hspace{4pt} & &\{X \colon c.f(a,c) \triangleq b.f(b,c) \} \spsys \emptyset \spsys \emptyset \spsys \id \\
\;\Lra_{\text{\sf Abs-T}}
    & \hspace{4pt} & &\{Y \colon f(a,b) \triangleq f(b,c) \} \spsys \emptyset \spsys \emptyset \spsys \{X\mapsto b.Y \} \\
\;\Lra_{\text{\sf Dec-T}}
    & \hspace{4pt} & &\{Z \colon (a,b) \triangleq (b,c) \} \spsys \emptyset \spsys \emptyset \spsys \{X \mapsto b.f(Z), \ldots\} \\
\;\Lra_{\text{\sf Dec-H}}
    & \hspace{4pt} & &\{Z_1 \colon a \triangleq b, \, Z_2 \colon b \triangleq c \} \spsys \emptyset \spsys \emptyset \spsys
    \\ & \hspace{4pt} & &\{ X\mapsto b.f(Z_1,Z_2), \ldots\} \\
\;\Lra_{\text{\sf Sol-T}}^2
    & \hspace{4pt} & &\emptyset \spsys \{Z_1 \colon a \triangleq b, \, Z_2 \colon b \triangleq c \} \spsys \{ c\# Z_1, a\# Z_2  \} \spsys
    \\ & \hspace{4pt} & &\{ X\mapsto b.f(Z_1,Z_2), \ldots\} \\
\;\Lra_{\text{\sf{Mer}}}
    & \hspace{4pt} & &\emptyset \spsys \{Z_1 \colon a \triangleq b \} \spsys \{ c\# Z_1 \} \spsys
    \\ & \hspace{4pt} & &\{ X\mapsto b.f(Z_1,\swap{b}{a}\swap{c}{b} \permef Z_1), \ldots \} \\
\;\Lra_{\text{\sf{Nar-T}}}
    & \hspace{4pt} & &\emptyset \spsys \{z \colon a \triangleq b \} \spsys \{ c\# z \} \spsys
    \\ & \hspace{4pt} & &\{ X\mapsto b.f(z,\swap{b}{a}\swap{c}{b}\permef z),  \ldots \}.
\end{alignat*}

\end{example}

\begin{example}
    Let $s=f(a,b,b,a)$ and $t=f(Y, \swap{a}{b} \permef Y)$ with $A=\{a,b\}$ and $\nabla=\emptyset$. Among the final states the algorithm stops with are the following:
    \begin{align*}
    S_1   := {} & \emptyset \spsys \{ Z_1 \colon a \triangleq \varepsilon,\ Z_2 \colon b \triangleq Y \} \spsys \{b \# Z_1\} \spsys 
    \\ & \{ X \mapsto f(Z_1, \swap{a}{b} \permef Z_1, Z_2, \swap{a}{b} \permef Z_2), \ldots \} \\
    S_2   := {} & \emptyset \spsys \{ Z_1 \colon a \triangleq \varepsilon,\ Z_2 \colon b \triangleq Y \} \spsys \{b \# Z_1\} \spsys
    \\ & \{ X \mapsto f(Z_1, Z_2, \swap{a}{b} \permef Z_1, \swap{a}{b} \permef Z_2), \ldots  \} \\
    S_3  := {} & \emptyset \spsys \{ Z_1 \colon a \triangleq \varepsilon,\ Z_2 \colon b \triangleq Y,\ Z_3 \colon b \triangleq \swap{a}{b} \permef Y \} \spsys
    \\ & \{b \# Z_1\} \spsys \{X \mapsto f(Z_1, Z_2, Z_3, Z_1), \ldots  \} \\
    S_4  := {} & \emptyset \spsys \{ Z_1 \colon a \triangleq Y,\ Z_2 \colon b \triangleq \varepsilon\} \spsys \{a \# Z_2\} \spsys
    \\ & \{ X \mapsto f(Z_1, Z_2, \swap{a}{b} \permef Z_1, \swap{a}{b} \permef Z_2),\ldots   \} \\
    S_5  := {} & \emptyset \spsys \{ Z_1 \colon a \triangleq Y,\ Z_2 \colon b \triangleq \varepsilon \} \spsys \{a \# Z_2\} \spsys
    \\ & \{ X \mapsto f(Z_1, \swap{a}{b} \permef Z_1, Z_2, \swap{a}{b} \permef Z_2),\ldots  \} \\
    S_6   := {} & \emptyset \spsys \{ Z_1 \colon a \triangleq Y,\ Z_2 \colon b \triangleq \varepsilon,\ Z_3 \colon a \triangleq \swap{a}{b} \permef Y \} \spsys
    \\ & \{a \# Z_2\} \spsys \{ X \mapsto f(Z_1, Z_2, Z_2, Z_3),\ldots  \} \\
    S_7   := {} & \emptyset \spsys \{ Z_1 \colon a \triangleq \varepsilon,\ Z_2 \colon \varepsilon \triangleq Y \} \spsys \{b \# Z_1\} \spsys
    \\ & \{ X \mapsto f(Z_1, \swap{a}{b} \permef Z_1, Z_2, \swap{a}{b} \permef Z_2, \swap{a}{b} \permef Z_1, Z_1),\ldots  \} 
    \end{align*}
    The terms-in-context obtained from the states $S_1$, $S_2$, $S_4$, $S_5$ and $S_7$ are all pairwise equi-general. $S_3$ and $S_6$ are strictly more general than the others. Moreover, $S_1$, $S_2$, $S_4$, $S_5$ and $S_7$ are also equi-general to the input term-in-context ${T=\np{\emptyset}{f(Y, \swap{a}{b} \permef Y)=t}}$.
    For instance, consider the term-in-context $U=\np{\{a \# Z_2\}}{f(Z_1, Z_2, \swap{a}{b} \permef Z_1, \swap{a}{b} \permef Z_2)=u}$, obtained from $S_4$.
Applying Definition \ref{def:more-general-relation}, we get:
\begin{align*}
	T\preceq U,&\text{ since } \sigma=\{Y\mapsto (Z_1, Z_2) \}\text{ respects } \emptyset
	\\
	&\text{ and } \,\emptyset\sigma\subseteq \{a \# Z_2\}\text{ and }\{a \# Z_2\}\vdash t\sigma\approx u.
	\\
	U\preceq T,&\text{ since } \sigma=\{Z_1\mapsto Y, Z_2\mapsto \varepsilon \}\text{ respects } \{a \# Z_2\}
	\\
	&\text{ and } \,\{a \# Z_2\}\sigma\subseteq \emptyset\text{ and }\emptyset\vdash u\sigma\approx t.
\end{align*}
\end{example}

Now we address important properties of the algorithm: termination, soundness, and completeness.

The complexity measure $\mu: State\rightarrow(\N,\N,\N,\N,\N)$, where $State$ denotes some state $P;\, S;\, \Gamma;\, \sigma$, is defined component-wise in the following manner: 
\begin{enumerate}
	\item Sum of all function symbol occurrences and all abstraction occurrences in all the AUPs of $P$.
	\item Sum of the squared problem lengths defined by $\{ \len{\tilde{s},\tilde{q}}^2 \mid X\colon \tilde s \triangleq \tilde q \in P\}$.
	\item Number of AUPs in $P$.
	\item Number of AUPs in $S$.
	\item Sum of hedge variable occurrences in all the AUPs of $S$.
\end{enumerate}
Measures are ordered lexicographically. The ordering, denoted by $>$, is well-founded.
By the termination theorem, it directly follows that $\unau$ terminates on any input.
\begin{theorem}[Termination]\label{thm:termination}
Let $P;S; \Gamma; \sigma$ be an arbitrary state. There are finitely many possible rule applications $P; S; \Gamma; \sigma\overset{\mathsf{Appl}_i}{\Lra} P_i; S_i; \Gamma_i; \sigma_i, 1\le i\le n$ and $\mu(P; S; \Gamma;\sigma) > \mu(P_i; S_i; \Gamma_i; \sigma_i) $ holds for all $1\leq i\leq n$.
\end{theorem}
\begin{proof}
Assuming that $P$ and $S$ are finite, each rule can only be applied in finitely many ways to $P \spsys S \spsys \Gamma \spsys \sigma$. Furthermore, each rule strictly reduces the measure $\mu$.
\end{proof}

\begin{lemma}\label{lem:state-invariant}
Given a state $P;\, S;\, \Gamma;\, \sigma$ and a freshness context $\nabla$, everything based on a finite set of atoms~$A$, such that, for all $\x\colon \s \triangleq \q \in  P \cup S$ holds $\np{\Gamma}{\x\sigma}$ is an $A$-based generalization of $\np{\nabla}{\s}$ and $\np{\nabla}{\q}$.

If $P; S; \Gamma;\sigma \Lra P'; S'; \Gamma'; \sigma'$ is a transformation by a $\unau$ rule,
then, for all $\x\colon \s \triangleq\nobreak \q \in  P \cup S\cup P'\cup S'$,  $\np{\Gamma'}{\x\sigma'}$ is an $A$-based generalization of $\np{\nabla}{\s}$ and $\np{\nabla}{\q}$.
\end{lemma}
The proof of Lemma \ref{lem:state-invariant} uses two substitutions that can be obtained from any set $M$ of AUPs:
\begin{align*}
\substils{M} ::= &\;
\{\x\mapsto \s \mid \x\colon \s \triangleq \q \in M \}
\text{ and }
\\
\substirs{M} ::= &\;
\{\x\mapsto \q \mid \x\colon \s \triangleq \q \in M \}.
\end{align*}
\begin{proof}
Inspecting the rules, it is easy to see that each of them fulfills the invariant.
In particular, given a transformation $P;\, S;\, \Gamma;\, \sigma \Lra P';\, S';\, \Gamma';\, \sigma'$, every rule fulfills:
\begin{itemize}
    \item $\s\approx \x\sigma'\substils{P'\cup S'}$ and $\q\approx \x\sigma'\substirs{P'\cup S'}$ hold for all $\x\colon \s \triangleq \q \in  P \cup S$.
\end{itemize}
The second part, namely that for all $\x\colon \s \triangleq \q \in  P' \cup S' \text{ holds } \s\approx \x\sigma'\substils{P'\cup S'} \text{ and } \q\approx \x\sigma'\substirs{P'\cup S'}$, is trivial since $\sigma'$ does not contain assignments for new generalization variables: it simply reduces to $\s\approx \x\substils{P'\cup S'} \text{ and } \q\approx \x\substirs{P'\cup S'}$.
\end{proof}

The soundness theorem states that $\unau$ computes $A$-based generalizations of the given $A$-based input terms-in-context.
\begin{theorem}[Soundness]
	\label{thm:soundness}
	Given two hedges or terms $\s,\q$ and a freshness context $\nabla$,
	all based on a finite set of atoms $A$. If
	$\{X: \s\triangleq \q\};\, \emptyset;\, \emptyset;\, \id \overset{+}{\Lra} \emptyset;\,S; \Gamma;\, \sigma$
	is a derivation obtained by an execution of $\unau$,
	then $\np{\Gamma}{X\sigma}$ is an $A$-based generalization of
	$\np{\nabla}{\s}$ and $\np{\nabla}{\q}$.
\end{theorem}
\begin{proof}
Follows from Lemma \ref{lem:state-invariant}, by induction on the derivation length.
\end{proof}

The completeness theorem states that, given two $A$-based terms-in-context and an arbitrary $A$-based generalization $G$ of them, $\unau$ computes an $A$-based generalization of the given terms-in-context that is less (or equi-) general than $G$.

\begin{theorem}[Completeness]\label{thm:completeness}
	Given two hedges or terms $\s,\q$ and a freshness context $\nabla$,
	all based on a finite set of atoms $A$. If $G$ is
	an $A$-based generalization of $\np{\nabla}{\s}$ and
	$\np{\nabla}{\q}$, then there exists a derivation
	$
	\{X: \s\triangleq \q\};\, \emptyset;\, \emptyset;\, \id \overset{+}{\Lra}
	\emptyset;\,S;  \Gamma;\, \sigma
	$
	obtained by an execution of $\unau$, such that
	$G\preceq \np{\Gamma}{X\sigma}$.
\end{theorem}
\begin{proof}[Proof sketch] The idea is similar to the completeness proofs from \cite{DBLP:conf/rta/BaumgartnerKLV15,DBLP:journals/jar/KutsiaLV14} and proceeds by induction over the structure of the given generalization.
\end{proof}

\begin{example}\label{ex:clone-appl-unau}
Consider the variadic nominal representation of code pieces from Figure \ref{fig:sw-clone-hedges}. Let $s$ be the original code and $t$ be its type-3 clone. The set of atoms is $A=\atoms(s)\cup\atoms(t) = \{n,sum,\allowbreak prod,\allowbreak i,\allowbreak a,\allowbreak s,\allowbreak p,j\}$. The input freshness context is empty.
The algorithm renames abstractions, in accordance to the equivalence relation $\approx$, and progresses down the terms that represent the abstract syntax trees:
\[\arraycolsep=1pt
\begin{array}{ll}
	& \{X \colon s \triangleq t \} \spsys \emptyset \spsys \emptyset \spsys \id \\
	\Lra_{\text{\sf Abs-T}}^3
	& \{X' \colon s' \triangleq t' \} \spsys \emptyset \spsys \emptyset \spsys \{X\mapsto a.s.p.X' \} \\
	\Lra_{\text{\sf Dec-T}}
	& \{X''\colon \tilde s \triangleq \tilde q \} \spsys \emptyset \spsys \emptyset ; \{X\mapsto a.s.p.\text{\normalfont sumProd($X''$)},\dots \}
\end{array}
\]
{
\setlength{\tabcolsep}{3pt}
\begin{tabular}{rl}
where $\tilde s$ is
&\normalfont input(type(int), $a$), returnType(void),\\
&\normalfont   {=}(type(float), $s$, 0.0), {=}(type(float), $p$, 1.0),\\
&\normalfont   $i.$for({=}(type(int),$i$,1), {$\leq$}($i,a$), \incmt($i$),\\
&\normalfont \qquad {=}($s$, +($s,i$)), {=}($p$, *($p,i$)), foo($s,p$))\\[1mm]
and $\tilde q$ is
&\normalfont input(type(int), $a$), returnType(void),\\
&\normalfont   {=}(type(double), $s$, 0.0), {=}(type(double), $p$, 1.0),\\
&\normalfont   $j.$for({=}(type(int),$j$,1), {$\leq$}($j,a$), \incmt($j$),\\
&\normalfont \qquad {=}($s$, +($s,j$)), foo($j,s,p,a$)).
\end{tabular}
}

\medskip
Consecutive applications of the rule \irule{Dec-H} generate various branches. Here we illustrate only one, which leads to an lgg interesting from the point of view of clones.

Applying \irule{Dec-H} four times, using $\len{\tilde s_1}=\len{\tilde q_1}=1$ where $\tilde s = \tilde s_1, \tilde s_2$ and $\tilde q = \tilde q_1, \tilde q_2$, emits the state $Q$: 
\[\arraycolsep=1pt
\begin{array}{ll}
	Q:=
	& \{Y_1\colon \text{\normalfont input(type(int), $a$)} \triangleq \text{\normalfont input(type(int), $a$)},
	\\ & \phantom{\{}Y_2\colon \text{\normalfont returnType(void)} \triangleq \text{\normalfont returnType(void)},
	\\ & \phantom{\{}Y_3\colon \text{\normalfont {=}(type(float),\:$s$,\:0.0)} \triangleq \text{\normalfont {=}(type(double),\:$s$,\:0.0)},
	\\ & \phantom{\{}Y_4\colon \text{\normalfont {=}(type(float),\:$p$,\:1.0)} \triangleq \text{\normalfont {=}(type(double),\:$p$,\:1.0)},
	\\ & \phantom{\{}Y_5\colon s'' \triangleq t''
	\} \spsys 
	\\ & \;\emptyset ;\, \emptyset ;\, \{X\mapsto a.s.p.\text{\normalfont sumProd($Y_1,Y_2,Y_3,Y_4,Y_5$)},\dots \}
\end{array}
\]
{
\setlength{\tabcolsep}{3pt}
\begin{tabular}{rl}
	where $s''$ is
	&\normalfont   $i.$for({=}(type(int),$i$,1), {$\leq$}($i,a$), \incmt($i$),\\
	&\normalfont \qquad {=}($s$, +($s,i$)), {=}($p$, *($p,i$)), foo($s,p$))\\[1mm]
	and $t''$ is
	&\normalfont   $j.$for({=}(type(int),$j$,1), {$\leq$}($j,a$), \incmt($j$),\\
	&\normalfont \qquad {=}($s$, +($s,j$)), foo($j,s,p,a$)).
\end{tabular}
}

\medskip
From the state $Q$, the branch continues as follows:
\begin{alignat*}{4}
  \Lra_{\text{\sf{Dec-T}},}^+ & & & {_\text{\sf{Dec-H,\,Tri-T}}^{\phantom{+}}} 
    \\& & &
    \{Y_3'\colon \text{\normalfont float} \triangleq \text{\normalfont double},\; Y_4'\colon \text{\normalfont float} \triangleq \text{\normalfont double},
	\\ & & & \phantom{\{}Y_5\colon s'' \triangleq t''
	\} \spsys \emptyset \spsys \emptyset \spsys 
	\\ & & & \{X\mapsto a.s.p.\text{\normalfont sumProd(input(type(int),\,$a$)},\:
    \\ & & & \phantom{\{X\mapsto a.s.p.} \quad \text{\normalfont returnType(void),}
	\\ & & & \phantom{\{X\mapsto a.s.p.} \quad \text{\normalfont {=}(type($Y_3'$),\,$s$,\,0.0),}  
    \\ & & & \phantom{\{X\mapsto a.s.p.} \quad \text{\normalfont{=}(type($Y_4'$),\,$p$,\,1.0),$Y_5$)},\dots \} \\
 \Lra_{\text{\sf{Sol-T}},} & & & {_\text{\sf{Mer,\,Nar-T}}} \\
    & & & \{Y_5\colon s'' \triangleq t''
	\} \spsys \{z\colon \text{\normalfont float} \triangleq \text{\normalfont double}\} \spsys
    \\ & & & \{b\#z\mid b\in A\} \spsys 
    \\ & & & \{X\mapsto a.s.p.\text{\normalfont sumProd(input(type(int),\,$a$)},\:
    \\ & & & \phantom{\{X\mapsto a.s.p.} \quad \text{\normalfont returnType(void),}
	\\ & & & \phantom{\{X\mapsto a.s.p.} \quad \text{\normalfont {=}(type($z$),\,$s$,\,0.0),}  
    \\ & & & \phantom{\{X\mapsto a.s.p.} \quad \text{\normalfont{=}(type($z$),\,$p$,\,1.0)},Y_5),\dots \}  \\
    \Lra_{\text{\sf{Abs-T}},} & & & {_\text{\sf{Dec-T}}} 
    \\ & & & \{Y_5'\colon \tilde{s}' \triangleq \tilde{q}'
	\} \spsys \{z\colon \text{\normalfont float} \triangleq \text{\normalfont double}\} \spsys 
    \\ & & & \{b\#z\mid b\in A\} \spsys 
   \\ & & & \{X\mapsto a.s.p.\text{\normalfont sumProd(input(type(int),\,$a$)},\:
    \\ & & & \phantom{\{X\mapsto a.s.p.} \quad \text{\normalfont returnType(void),}
	\\ & & & \phantom{\{X\mapsto a.s.p.} \quad \text{\normalfont {=}(type($z$),\,$s$,\,0.0),}  
    \\ & & & \phantom{\{X\mapsto a.s.p.} \quad \text{\normalfont{=}(type($z$),\,$p$,\,1.0)}, 
    \\ & & & \phantom{\{X\mapsto a.s.p.}  \quad i.\text{\normalfont for($Y_5'$)}),\dots \}  
\end{alignat*}

{
\setlength{\tabcolsep}{3pt}
\begin{supertabular}{rl}
	where $\tilde{s}'$ is
	&\normalfont {=}(type(int),$i$,1), {$\leq$}($i,a$), \incmt($i$),\\
	&\normalfont {=}($s$, +($s,i$)),\\
	&\normalfont {=}($p$, *($p,i$)), foo($s,p$)\\[1mm]
	and $\tilde{q}'$ is
	&\normalfont {=}(type(int),$i$,1), {$\leq$}($i,a$), \incmt($i$),\\
	&\normalfont {=}($s$, +($s,i$)), foo($i,s,p,a$).\\
\end{supertabular}
}\vspace{3mm}

Applying the hedge decomposition rule \irule{Dec-H} several times to the problem $Y_5':\tilde{s}'\triangleq \tilde{q}'$, we get:
\begin{alignat*}{4}
    & & & \text{\normalfont{$\{Z_1\colon${=}(type(int),$i$,1) $\triangleq$ {=}(type(int),$i$,1)}},
   \\ & & & \text{\normalfont{\phantom{$\{$}$Z_2\colon${$\leq$}($i,a$) $\triangleq$ {$\leq$}($i,a$)}}, 
   \\ & & & \text{\normalfont{\phantom{$\{$}$Z_3\colon$\incmt($i$) $\triangleq$ \incmt($i$),}} 
   \\ & & & \text{\normalfont{\phantom{$\{$}$Z_4\colon${=}($s$,\,+($s,i$)) $\triangleq$ {=}($s$,\,+($s,i$)),}} 
   \\ & & & \text{\normalfont{\phantom{$\{$}$Z_5\colon${=}($p$,\,*($p,i$)) $\triangleq \varepsilon$,}} 
	  \\ & & & \text{\normalfont{\phantom{$\{$}$Z_6\colon$foo($s,p$) $\triangleq$ foo($i,s,p,a$)$\}\spsys$}}
     \\ & & & 
     \{z\colon \text{\normalfont float} \triangleq \text{\normalfont double}\} \spsys 
     \\ & & & \{b\#z\mid b\in A\} \spsys 
    \\ & & & \{X\mapsto a.s.p.\text{\normalfont sumProd(input(type(int),\,$a$)},\:
    \\ & & & \phantom{\{X\mapsto a.s.p.} \quad \text{\normalfont returnType(void),}
	\\ & & & \phantom{\{X\mapsto a.s.p.} \quad \text{\normalfont {=}(type($z$),\,$s$,\,0.0),}  
    \\ & & & \phantom{\{X\mapsto a.s.p.} \quad \text{\normalfont{=}(type($z$),\,$p$,\,1.0)}, 
    \\ & & & \phantom{\{X\mapsto a.s.p.}  \quad i.\text{\normalfont for($Z_1,Z_2,Z_3,Z_4,Z_5,Z_6$)}),\dots \} 
\end{alignat*}

The first four AUPs can be transformed by a sequence of \irule{Dec-T}, \irule{Dec-H},  \irule{Tri-T}, and \irule{Tri-H} rules, and the fifth one with \irule{Sol-H}, obtaining
\begin{alignat*}{4}
    & & &  \{ Z_6\colon \text{\normalfont foo($s,p$) $\triangleq$ foo($i,s,p,a$)\} \spsys }
     \\ & & & 
     \{z\colon \text{\normalfont float} \triangleq \text{\normalfont double}, Z_5 \colon \text{\normalfont {=} ($p$,\,*($p,i$)) }\triangleq \varepsilon \} \spsys 
     \\ & & & \{b\#z\mid b\in A\} \cup \{b \# Z_5 \mid b \in A\setminus \{ p,i\} \} \spsys 
    \\ & & & \{X\mapsto a.s.p.\text{\normalfont sumProd(input(type(int),\,$a$)},\:
    \\ & & & \phantom{\{X\mapsto a.s.p.} \quad \text{\normalfont returnType(void),}
	\\ & & & \phantom{\{X\mapsto a.s.p.} \quad \text{\normalfont {=}(type($z$),\,$s$,\,0.0),}  
    \\ & & & \phantom{\{X\mapsto a.s.p.} \quad \text{\normalfont{=}(type($z$),\,$p$,\,1.0)}, 
    \\ & & & \phantom{\{X\mapsto a.s.p.}  \quad i.\text{\normalfont for({=}(type(int),$i$,1),} 
     \\ & & & \phantom{\{X\mapsto a.s.p. \quad i.\text{\normalfont for(}} \text{\normalfont {$\leq$}($i,a$),} 
     \\ & & & \phantom{\{X\mapsto a.s.p. \quad i.\text{\normalfont for(}} \text{\normalfont \incmt($i$),} 
    \\ & & & \phantom{\{X\mapsto a.s.p. \quad i.\text{\normalfont for(}} \text{\normalfont {=}($s$,\,+($s,i$)),} 
    \\ & & & \phantom{\{X\mapsto a.s.p. \quad i.\text{\normalfont for(}} Z_5,Z_6)),\dots \}
\end{alignat*}

In the final stage of this derivation, the AUP $Z_6\colon \text{\normalfont foo($s,p$) $\triangleq$ foo($i,s,p,a$)}$ will be processed by a sequence of \irule{Dec-T}, \irule{Dec-H},  \irule{Sol-H}, \irule{Tri-T}, \irule{Tri-H}, and \irule{Mer} rules, which will give the final state:
\begin{alignat*}{4}
    & & &  \emptyset \spsys 
     \{z\colon \text{\normalfont float} \triangleq \text{\normalfont double}, Z_5 \colon \text{\normalfont {=} ($p$,\,*($p,i$)) }\triangleq \varepsilon, V \colon \varepsilon \triangleq i \}  \spsys 
     \\ & & & \{b\#z\mid b\in A\} \cup \{b \# Z_5 \mid b \in A\setminus \{ p,i\} \} \cup {} 
    \\ & & & \quad \{b\# V\mid b\in A\setminus \{i\} \} \spsys 
    \\ & & & \{X\mapsto a.s.p.\text{\normalfont sumProd(input(type(int),\,$a$)},\:
    \\ & & & \phantom{\{X\mapsto a.s.p.} \quad \text{\normalfont returnType(void),}
	\\ & & & \phantom{\{X\mapsto a.s.p.} \quad \text{\normalfont {=}(type($z$),\,$s$,\,0.0),}  
    \\ & & & \phantom{\{X\mapsto a.s.p.} \quad \text{\normalfont{=}(type($z$),\,$p$,\,1.0)}, 
    \\ & & & \phantom{\{X\mapsto a.s.p.}  \quad i.\text{\normalfont for({=}(type(int),$i$,1),} 
     \\ & & & \phantom{\{X\mapsto a.s.p. \quad i.\text{\normalfont for(}} \text{\normalfont {$\leq$}($i,a$),} 
     \\ & & & \phantom{\{X\mapsto a.s.p. \quad i.\text{\normalfont for(}} \text{\normalfont \incmt($i$),} 
    \\ & & & \phantom{\{X\mapsto a.s.p. \quad i.\text{\normalfont for(}} \text{\normalfont {=}($s$,\,+($s,i$)),} 
    \\ & & & \phantom{\{X\mapsto a.s.p. \quad i.\text{\normalfont for(}} Z_5,
    \\ & & & \phantom{\{X\mapsto a.s.p. \quad i.\text{\normalfont for(}} \text{\normalfont {=}foo($V,s,p,\swap{a}{i}\permef V$)} )),\dots \}
\end{alignat*}

The final state contains several pieces of interesting information. First, it reveals the similarity between the given code fragments (represented by $s$ and $t$) in the form of the generalization term assigned to $X$ in the final substitution. Notably, the original differences between the two code fragments, such as variations in the names of arguments and local variables, are effectively handled in the computed generalization by renaming and identifying bound atoms (thus $a,s,p,i$ in the answer), thereby enhancing the quality of the detected similarity.\footnote{For the same reason, comparing $s$ with its type-2 clone results in a generalization that shows the two fragments are almost identical modulo bound atom names, thus providing more accurate insight into the nature of the clone.} 

Further, the store $\{z\colon \text{\normalfont float} \triangleq \text{\normalfont double}, Z_5 \colon \text{\normalfont {=} ($p$,\,*($p,i$))}\triangleq\nobreak \varepsilon, \allowbreak V \colon \varepsilon \triangleq i \}$ provides the ``difference measure'' between the original code fragments (similar to, e.g., edit distance) and illustrates how each of them can be reconstructed from the generalization: replacing $z$ by float, $Z_5$ by $\text{\normalfont {=} ($p$,\,*($p,i$))}$, and $V$ by the empty hedge yields $s$, while replacing $z$ by double, eliminating $Z_5$ by replacing it with the empty hedge, and replacing $V$ by $i$ yields $t$. This information can be useful for code reformatters. Moreover, reusing $V$ in two different places in the generalization, with the suspension $\swap{a}{i}\permef V$ in one of them, precisely captures how these two positions differ from each other in the original code: by a swap of atoms $a$ and $i$. Further potentially useful information is contained in the freshness constraints, showing which atoms are forbidden in the missing places (i.e., which local or argument variables do not appear there).

Despite all these advantages, the algorithm has some serious limitations, particularly concerning efficiency: the nondeterminism involved is too high to allow practical use, as the next example illustrates. Moreover, the algorithm offers little control over how similar parts are selected in the input hedges, instead relying on an exhaustive brute-force search. These drawbacks will be addressed in the following section.
\end{example}

\begin{example}\label{ex:many-lggs}
Let $s=f(a,b,b,a)$ and $t=f(c, c)$ with $A=\{a,b,c\}$ and $\nabla=\emptyset$. The following final states emit pairwise incomparable lggs:
\begin{align*}
S_1   := {} & \emptyset \spsys \{ y \colon a \triangleq c,  Z \colon b \triangleq \varepsilon \} \spsys \{b \# y,\, a \# Z, c \# Z\} \spsys 
\\& \{ X \mapsto f(y, \swap{a}{b}\permef y, Z, \swap{a}{b}\permef Z), \ldots \}
\\
S_2   := {} & \emptyset \spsys \{ y \colon a \triangleq c,  Z \colon b \triangleq \varepsilon \} \spsys \{b \# y,\, a \# Z, c \# Z\} \spsys 
\\& \{ X \mapsto f(y, Z, \swap{a}{b}\permef y, \swap{a}{b}\permef Z), \ldots \}
\\
S_3   := {} & \emptyset \spsys \{ y \colon a \triangleq c,  Z \colon b \triangleq \varepsilon \} \spsys \{b \# y,\, a \# Z, c \# Z\} \spsys
\\& \{ X \mapsto f(y, Z, Z, y), \ldots \}
\\
S_4   := {} & \emptyset \spsys \{ y \colon b \triangleq c,  Z \colon a \triangleq \varepsilon \} \spsys \{a \# y,\, b \# Z, c \# Z\} \spsys 
\\& \{ X \mapsto f(Z, y, y, Z), \ldots \}
\\
S_5   := {} & \emptyset \spsys \{ y \colon b \triangleq c,  Z \colon a \triangleq \varepsilon \} \spsys \{a \# y,\, b \# Z, c \# Z\} \spsys 
\\& \{ X \mapsto f(Z, y, \swap{a}{b}\permef Z, \swap{a}{b}\permef y), \ldots \}
\\
S_6   := {} & \emptyset \spsys \{ y \colon b \triangleq c,  Z \colon a \triangleq \varepsilon \} \spsys \{a \# y,\, b \# Z, c \# Z\} \spsys 
\\& \{ X \mapsto f(Z, \swap{a}{b}\permef Z, y, \swap{a}{b}\permef y), \ldots \}
\\
S_7   := {} & \emptyset \spsys \{ y \colon a \triangleq c,  Y \colon \varepsilon \triangleq c,  Z \colon b \triangleq \varepsilon \} \spsys
\\&  \{b \# y,\, a \# Y, b \# Y,\, a \# Z, c \# Z\} \spsys
\\& \{ X \mapsto f(y, Y, Z, Z, \swap{a}{b}\permef Z), \ldots \}
\\
\dots\quad & \text{all positional permutations of } y, Y, Z, Z, Z \; \dots
\\
S_{27}   := {} & \emptyset \spsys \{ Y \colon \varepsilon \triangleq c,  Z \colon b \triangleq \varepsilon \} \spsys \{a \# Y, b \# Y,\, a \# Z, c \# Z\} \spsys
\\& \{ X \mapsto f(Y, Y, \swap{a}{b}\permef Z, Z, Z, \swap{a}{b}\permef Z), \ldots \}
\\
\dots\quad & \text{all positional permutations of } Y, Y, Z, Z, Z, Z \; \dots
\end{align*}
There are ${4 \choose 2}+{5 \choose 1}{4 \choose 1}+{6 \choose 2}=41$ lggs. Note that $S_{7}$ is an lgg because of the freshness constraints.
Without them, one could take, for instance $\sigma=\{Z\mapsto\varepsilon, \, Y\mapsto(Z,Z,y)\}$ to obtain the generalization of $S_3$. However, the assignment $Y\mapsto(Z,Z,y)$ does not respect $\{a \# Y, b \# Y\}$.
\end{example}

Example \ref{ex:many-lggs} illustrates that the set of lggs might become very large in the general case. One of the reasons is that hedge suspensions appear consecutively and only represent single terms. This counterintuitive behavior is addressed in Section~\ref{sec:au-rigid}.

\begin{rem}
    Using binders helps identify similarities modulo bound atom names, which improves the precision of clone detection. However, this can sometimes backfire when the terms to be generalized have different abstraction lengths. For example, compare the generalization $\np{\{a\#x, c\#x, d\#x\}}{a.b.f(a,x)}$ of $a.b.f(a,b)$ and $c.d.f(c,e)$ (the same abstraction lengths) to the generalization $\np{\{b\#x, c\#x\}}{a.x}$ of terms $a.b.f(a,b)$ and $c.f(c,e)$ (different abstraction lengths). The second one does not provide information about the structural similarity. 

Such mismatches can occur when the procedures being compared have different numbers of arguments or local variables. In these cases, if we still wish to generalize the code fragments, we can use an alternative encoding that ``frees'' the previously bound atoms (or some of them, making the number of abstractions equal). For instance, by encoding the original terms as $f(a,b)$ and $f(c,d)$ instead, we obtain the generalization $\np{\{b\#x, d\#x\}}{f(x, \swap{d}{c}\swap{b}{a}\permef x}$, which captures more structural similarity between the fragments.
\end{rem}

\section{Rigid Variadic Nominal Anti-Unification ($\unaur$)}\label{sec:au-rigid}
This section extends the idea of rigid generalizations from~\cite{DBLP:journals/jar/KutsiaLV14} to the nominal setting where bound atoms and permutations become involved.

Let $w_1$ and $w_2$ be sequences of symbols (words). The sequence $\alignsym_1\alignpos{i_1}{j_1}\cdots\alignsym_n\alignpos{i_n}{j_n}$, $n\geq 0$, is an \emph{alignment} of $w_1$ and $w_2$ if
\begin{itemize}
\item $0 < i_1 < \dots < i_n \leq \len{w_1}$ and $0 < j_1 < \dots < j_n \leq \len{w_2}$, and
\item $\alignsym_k=w_1|_{i_k}=w_2|_{j_k}$ and $\alignsym_k\notin \pmb{x}\cup\pmb{X}$, for all $1\leq k\leq n$.
\end{itemize}
A \emph{rigidity function}, denoted by $\R$, is a function that computes a set of alignments of any pair of words.
The function $\hseq$ returns the word of head symbols of a given hedge or individual term.
For instance, $\hseq(a,\: a,\: b,\: b.g(b),\: c,\: g(a),\: a.g(a),\: c)=aab.cg.c$.
A common choice for $\R$ is simply computing the set of all longest common subsequence alignments (see~\cite{DBLP:journals/jar/KutsiaLV14}), denoted by $\Rlcs$. An example that illustrates such a rigidity function is $\R(ab.dax,\: b.adx)=\{ b\alignpos{2}{1}.\alignpos{3}{2}a\alignpos{5}{3},\: b\alignpos{2}{1}.\alignpos{3}{2}d\alignpos{4}{4}\}$. Note that $x$ is not allowed to be part of an alignment.

\begin{definition}
	\label{def:r-generalization}
	Given a rigidity function $\R$ and two $A$-based terms-in-context $\np{\s}{\nabla}$ and $\np{\q}{\nabla'}$. An $A$-based generalization $\np{\Gamma}{\w}$ of $\np{\nabla}{\s}$ and $\np{\nabla'}{\q}$ is an $A$-based \emph{$\R$-generalization} if either $\R(\hseq(\s),\hseq(\q))=\emptyset$ and $\w$ is a
	hedge variable, or there exists an alignment
	$\alignsym_1\alignpos{i_1}{j_1}\cdots\alignsym_n\alignpos{i_n}{j_n}\in
	\R(\hseq(\s),\hseq(\q))$ such that the following conditions are
	fulfilled:
	\begin{enumerate}
		\item The term or hedge $\w$ does not contain pairs of consecutive suspensions. I.e., substrings of the form $\pi\permef\x,\,\pi'\permef\x'$ where $\x,\x'\in \pmb x\cup  \pmb X$
		are not allowed to appear in $\w$.
		\item If all the suspensions from the (possibly singleton or empty) hedge $\w$ are removed, then a hedge of $n$ terms $t_1,\dots,t_n$ remains, such that $\head{t_i}=\alignsym_i$, $1\leq k\leq n$.
		\item For every $1\le k\le n$, if $t_k$ is an application $t_k=\alignsym_k(\w_k)$, there exists a pair of
		hedges $\s_k$ and $\q_k$ such that
		$\s|_{i_k}=\alignsym_k(\s_k)$, $\q|_{j_k}=\alignsym_k(\q_k)$ and
		$\np{\Gamma}{\w_k}$ is an $A$-based $\R$\nobreakdash-generalization of $\np{\nabla}{\s_k}$ and $\np{\nabla'}{\q_k}$.
		\item For every $1\le k\le n$, if $t_k$ is an abstraction $t_k=a.\w_k$, there exists a pair of
		terms $\s_k$ and $\q_k$ such that
		$\s|_{i_k}=b.\s_k$, $\q|_{j_k}=b'.\q_k$ and
		$\np{\Gamma}{\w_k}$ is an $A$-based $\R$\nobreakdash-generalization of $\np{\nabla}{\swap{a}{b}\permef\s_k}$ and $\np{\nabla'}{\swap{a}{b'}\permef\q_k}$.
	\end{enumerate}
\end{definition}

As usual, an $A$-based $\R$\nobreakdash-lgg is an $A$-based $\R$\nobreakdash-generalization $G$ of two terms-in-context $T$ and $U$ such that there is no other $A$-based $\R$\nobreakdash-generalization of $T$ and $U$ that is strictly less general than $G$.

The rigid anti-unification problem is parametric in a finite set of atoms $A\subset \pmb{A}$ and a rigidity function $\R$:
\begin{description}
	\item[Given:]
	A rigidity function $\R$ and two $A$-based terms-in-context $\np{\nabla}{\w_1}$ and $\np{\nabla}{\w_2}$.
	\item[Find:]
	An $A$-based $R$-lgg of $\np{\nabla}{\w_1}$ and $\np{\nabla}{\w_2}$.
\end{description}

Given some number $i$, we denote by $i^\incmt$ and $i^\decmt$, respectively, the number $i+1$ and $i-1$.
The rules of $\unaur$ are obtained by replacing the rules \irule{Dec-H} and \irule{Sol-H} from Section \ref{sec:au-general} by the following ones:

\infrule{{Dec-R}}{Rigid Hedge Decomposition}
{
	\{
	X\colon \tilde s \triangleq \tilde q
	\} \dotcup P \spsys S \spsys \Gamma \spsys \sigma
	\Lra
		\{
	Z_0\colon
	\tilde s|_1^{i_1^\decmt}
	\triangleq \tilde q|_1^{j_1^\decmt}
	\}
	\\[1mm]
	\cup
	\{
	Y_k\colon \tilde s|_{i_k} \triangleq \tilde q|_{j_k} \mid 1\leq k\leq n
	\}
	\cup 
	\{
	Z_n\colon
	\tilde s|_{i_{n}^\incmt}^{\len{\tilde s}}
	\triangleq \tilde q|_{j_{n}^\incmt}^{\len{\tilde q}}
	\}
	\\
	\cup 
	\{
	Z_k\colon \tilde s|_{i_{k}^\incmt}^{i_{k+1}^\decmt} \triangleq \tilde q|_{j_{k}^\incmt}^{j_{k+1}^\decmt} \mid 1\leq k\leq n-1
	\} 
	\cup P \spsys S \spsys \\[1mm]
	\Gamma \spsys \sigma\{X\mapsto Z_0,Y_1,Z_1,\dots,Y_n,Z_n \}
	,}
[{where at least one of the hedges $\tilde s$, $\tilde q$ is not a singleton,
	and $\alignsym_1\alignpos{i_1}{j_1}\cdots\alignsym_n\alignpos{i_n}{j_n}$
	is a non-empty alignment from $\R(\hseq(\tilde s),\hseq(\tilde q))$.}]

\infrule{{Sol-R}}{Rigid Hedge Solving}
{
  \{
X\colon \tilde s \triangleq \tilde q
 \} \dotcup P \spsys S \spsys \Gamma \spsys \sigma
 \Lra
 P \spsys \{
X\colon \tilde s \triangleq \tilde q
 \} \cup S \spsys  \{ a\#X \mid a\in A, \nabla \vdash a\#r_1 \,\wedge\, \nabla \vdash a\#r_2\} \cup \Gamma \spsys \sigma
,}
[{if none of the other rules is applicable to $X\colon \tilde s \triangleq \tilde q$.
}]\vspace{3mm}

The algorithm $\unaur$ works in the same manner as $\unau$, i.e., the modified set of rules is applied exhaustively in all possible ways. \irule{Dec-R} involves branching. The algorithm is obviously terminating. The proof is equivalent to the one of Theorem \ref{thm:termination}, even the measure $\mu$ introduced there can be used as is. Soundness follows from the fact that the term decomposition is governed by the rigidity function. Completeness is a bit more involved and can be shown by induction on the structure of the given rigid generalization.

\begin{example}
	We revisit Example \ref{ex:simple-two-lggs} and consider the rigitidy function $\Rlcs$.
	Since consecutive suspensions are not allowed, $\np{\emptyset}{c.f(a,c)}$ and $\np{\emptyset}{b.f(b,c)}$ have only one $\{a,b,c\}$-based $\Rlcs$-lgg,
	namely $\np{\{b\#V, c\#V, a\#W, b\#W\}}{b.f(V,b,W)}$.
\end{example}

\begin{example}\label{ex:clone-detection-rigid}
	Consider the code clone example (Fig. \ref{fig:sw-clone-hedges}), just like in Example \ref{ex:clone-appl-unau}, and the rigidity function $\Rlcs$.
	Instead of consecutively applying \irule{Dec-H}, we have to apply 
	\irule{Dec-R}. The number of alternatives get drastically reduced and we get only
 two $A$-based $\Rlcs$-lggs: $\np{\nabla}{u}$ and $\np{\nabla'}{u'}$ where
\begin{center}
\setlength{\tabcolsep}{3pt}\normalfont
\begin{tabular}{rl}
$u=$&$a.s.p.$sumProd(input(type(int), $a$), returnType(void),\\
	&\qquad  {=}(type(x), $s$, 0.0),\\
	&\qquad  {=}(type(x), $p$, 1.0),\\
	&\qquad  $i.$for({=}(type(int),$i$,1), {$\leq$}($i,a$), \incmt($i$),\\
	&\qquad\qquad {=}($s$, +($s,i$)),\\
	&\qquad\qquad $X$,\\
	&\qquad\qquad foo($Y,s,p,\swap{a}{i}\permef Y$)))\\
$\nabla=$&$\{b\#x\mid b\in A\}\cup\{b\#X\mid b\in A\setminus\{p,i\}\}$\\
	&$\cup\{b\#Y\mid b\in A\setminus\{i\}\}$
\end{tabular}

\smallskip
\begin{tabular}{rl}
$u'=$&$a.s.p.$sumProd(input(type(int), $a$), returnType(void),\\
	&\qquad  {=}(type(x), $s$, 0.0),\\
	&\qquad  {=}(type(x), $p$, 1.0),\\
	&\qquad  $i.$for({=}(type(int),$i$,1), {$\leq$}($i,a$), \incmt($i$),\\
	&\qquad\qquad $X$,\\
	&\qquad\qquad {=}($y$, $z$),\\
	&\qquad\qquad foo($Y,s,p,\swap{a}{i}\permef Y$)))\\
$\nabla'=$&$\{b\#x\mid b\in A\}\cup\{b\#X\mid b\in A\setminus\{s,i\}\}$\\
	&$\cup\{b\#y\mid b\in A\setminus\{s,p\}\}\cup\{b\#z\mid b\in A\setminus\{s,p,i\}\}$\\
	&$\cup\{b\#Y\mid b\in A\setminus\{i\}\}$
\end{tabular}
\end{center}
Hence, $\np{\nabla}{u}$ is the same answer as in Example \ref{ex:clone-appl-unau}. $\np{\nabla'}{u'}$ also provides interesting information, namely, that the missing part in one fragment can be aligned with another part while still retaining information about some similarity, although this relation between codes is not as precise as in the first answer.

\end{example}

Rigidity functions introduce a flexible mechanism for controlling similarity. For example, one can define a rigidity function that returns longest common subsequences above a specified minimal length, filtering out code fragments that are too short to be considered meaningful clones. The rigidity function can be also made to return singleton sets, making the algorithm deterministic and eliminating branching. Since rigidity functions operate on symbol sequences, they can rely on efficient string-based algorithms for their computation, offering the advantage of combining the speed of string-based techniques with the precision of tree-based methods.

Moreover, rigidity functions make it easy to incorporate approximate, quantitative string-based methods into the computation of similarity. This can significantly enhance the al\-gorithm's practical applicability, as it allows the detection of identifiers with similar, though not identical, names, which are often overlooked by purely syntactic approaches. The simi\-larity information between such names can be learned using AI techniques, such as embedding-based models, learned edit distances, or character-level neural networks trained on large codebases. This opens the door to hybrid systems that combine symbolic structure-aware reasoning with data-driven name similarity, improving both precision and recall in clone detection.

\section{Allowing Consecutive Individual Variables}\label{sec:au-x}

Rigid generalizations prohibit adjacent variables, which helps reduce the search space. However, there are cases where allowing variables to appear next to each other does not introduce significant computational overhead and can, in fact, improve the precision of the resulting generalizations. This applies in particular when the hedges abstracted by a hedge variable have exactly the same length. In such cases, the hedge variable can be replaced with a sequence of individual variables, one for each element of the hedge, enabling a more fine-grained generalization.

To formally define rigid generalizations where consecutive individual variables are permitted, we need to slightly adjust Definition \ref{def:r-generalization}, item 1).
\begin{definition}\label{def:rxx-generalization}
Let $\R$ be a rigidity function and consider two $A$-based terms-in-context $\np{\s}{\nabla}$ and $\np{\q}{\nabla'}$. An $A$-based generalization $\np{\Gamma}{\w}$ of $\np{\nabla}{\s}$ and $\np{\nabla'}{\q}$ is an $A$-based \emph{$\R$-generalization} if either $\R(\hseq(\s),\hseq(\q))=\emptyset$ and $\w$ is a
hedge variable, or there exists an alignment
$\alignsym_1\alignpos{i_1}{j_1}\cdots\alignsym_n\alignpos{i_n}{j_n}\in
\R(\hseq(\s),\hseq(\q))$ such that the following conditions are
fulfilled:
\begin{enumerate}
	\item The term or hedge $\w$ does not contain pairs of consecutive suspensions where some hedge variable is involved. I.e., substrings of the form $\pi\permef\x,\,\pi'\permef\x'$ where either $\x\in \pmb X$ or $\x'\in \pmb X$ are not allowed to appear in $\w$.
	On the other hand, consecutive suspensions of the form $\pi_1\permef x_1,\dots,\pi_n\permef x_n$, where $x_1,\dots,x_n\in \pmb x$ are allowed.
	\item Same as in Def. \ref{def:r-generalization}.
	\item Same as in Def. \ref{def:r-generalization}.
	\item Same as in Def. \ref{def:r-generalization}.
\end{enumerate}
\end{definition}

To compute rigid generalizations under such a modification, we need an extra rule that allows to decompose solved AUPs with the hedges of the same lengths:
    
\infrule{{Nar-H}}{Hedge Narrowing}
{
	P \spsys
	\{
	X\colon t_1\dots t_n \triangleq s_1\dots s_n
	\} \dotcup S \spsys \Gamma \spsys \sigma
	\Lra\\
	P \spsys \{
	x_1\colon t_1 \triangleq s_1,\dots, x_n\colon t_n \triangleq s_n
	\}\cup S \spsys
	\\
	\Gamma\{X\mapsto x_1,\dots, x_n\} \spsys \sigma\{X\mapsto x_1,\dots, x_n\}
	.}\vspace{3mm}

The other rules remain unchanged. 

\begin{example}
    Rigid generalization of $a.b.f(a,b)$ and $a.b.f(b,a)$ gives $a.b.f(X)$, where $X$ generalizes the hedges $a,b$ and $b,a$. If we allow consecutive individual variables, we get a more precise generalization $a.b.f(x, \swap{a}{b}\permef x)$, which shows not only the fact that $f$ is applied to two individual term arguments in both input, but also that the second argument is obtained from the first one by swapping $a$ and $b$.
\end{example}

\section{Future work}
\label{sect:future}

Variadic nominal anti-unification considered in this paper works in the first-order setting. A drawback of first-order anti-unification methods is that they overlook similarities between two terms if their heads are different. If one takes a large term $t$ and puts it under two different symbols $f$ and $g$ as $f(t)$ and $g(t)$, their generalization will be just some variable $x$ and the common part under the head symbols will be ignored. A~way to address this problem is to bring higher-order variables in the language. The simplest such extension would be the introduction of function variables that stand only for function symbols. Rigid anti-unification, extended by such variables, can use a modified \irule{Sol-T} rule, which does not immediately send AUPs like $f(t) \triangleq g(t)$ to the store, but analyzes the arguments below $f$ and $g$. If their similarity is above a predefined threshold (e.g., the alignment computed by rigidity function is ``long enough'') , then $f$ and $g$ are generalized by some function variable and the process continues with their arguments. This is a relatively light-weight extension of variadic nominal generalization, whose formal properties and possible applications in program analysis are left for future work.

\section{Conclusion}\label{sec:conc}

We have presented a framework for variadic nominal anti-unification, designed to detect structural similarity in tree-like expressions, with a particular focus on software code clone detection. Our approach generalizes syntactic structures while preserving the semantics of bound variables, relying on the expressive capabilities of nominal terms and variadic syntax.

We introduced a general anti-unification algorithm and established its termination, soundness, and completeness. To address practical challenges such as computational overhead and limited control over generalization precision, we proposed an extension based on rigidity functions. This refinement enhances efficiency and offers greater flexibility, enabling the algorithm to be tuned for specific applications. In particular, rigidity functions support hybrid approaches that combine the speed of string-based similarity techniques with the structural precision of tree-based methods. Their design also facilitates the integration of approximate and learnable similarity measures, allowing for the incorporation of AI-driven techniques.

The framework further supports extensions that enable fine-grained correspondence by replacing hedge variables with sequences of individual variables when the underlying hedges are of equal length, increasing the precision of generalizations. Additionally, we discussed a potential refinement involving function variables, which would allow the detection of similarities that occur under different function symbols, broadening the algorithm's applicability to more diverse code patterns.
\bibliographystyle{IEEEtran}
\bibliography{antiunif}

\begin{thebibliography}{10}
\providecommand{\url}[1]{#1}
\csname url@samestyle\endcsname
\providecommand{\newblock}{\relax}
\providecommand{\bibinfo}[2]{#2}
\providecommand{\BIBentrySTDinterwordspacing}{\spaceskip=0pt\relax}
\providecommand{\BIBentryALTinterwordstretchfactor}{4}
\providecommand{\BIBentryALTinterwordspacing}{\spaceskip=\fontdimen2\font plus
\BIBentryALTinterwordstretchfactor\fontdimen3\font minus
  \fontdimen4\font\relax}
\providecommand{\BIBforeignlanguage}[2]{{%
\expandafter\ifx\csname l@#1\endcsname\relax
\typeout{** WARNING: IEEEtran.bst: No hyphenation pattern has been}%
\typeout{** loaded for the language `#1'. Using the pattern for}%
\typeout{** the default language instead.}%
\else
\language=\csname l@#1\endcsname
\fi
#2}}
\providecommand{\BIBdecl}{\relax}
\BIBdecl

\bibitem{Plotkin70}
G.~D. Plotkin, ``A note on inductive generalization,'' \emph{Machine Intell.},
  vol.~5, no.~1, pp. 153--163, 1970.

\bibitem{Reynolds70}
J.~C. Reynolds, ``Transformational systems and the algebraic structure of
  atomic formulas,'' \emph{Machine Intell.}, vol.~5, no.~1, pp. 135--151, 1970.

\bibitem{DBLP:conf/ijcai/CernaK23}
\BIBentryALTinterwordspacing
D.~M. Cerna and T.~Kutsia, ``Anti-unification and generalization: {A} survey,''
  in \emph{Proceedings of the Thirty-Second International Joint Conference on
  Artificial Intelligence, {IJCAI} 2023, 19th-25th August 2023, Macao, SAR,
  China}.\hskip 1em plus 0.5em minus 0.4em\relax ijcai.org, 2023, pp.
  6563--6573. [Online]. Available:
  \url{https://doi.org/10.24963/ijcai.2023/736}
\BIBentrySTDinterwordspacing

\bibitem{DBLP:journals/pacmpl/LiZDZW24}
\BIBentryALTinterwordspacing
X.~Li, X.~Zhou, R.~Dong, Y.~Zhang, and X.~Wang, ``Efficient bottom-up synthesis
  for programs with local variables,'' \emph{Proc. {ACM} Program. Lang.},
  vol.~8, no. {POPL}, pp. 1540--1568, 2024. [Online]. Available:
  \url{https://doi.org/10.1145/3632894}
\BIBentrySTDinterwordspacing

\bibitem{DBLP:conf/icse/ZhengS24}
\BIBentryALTinterwordspacing
D.~Zheng and K.~Sen, ``Dynamic inference of likely symbolic tensor shapes in
  python machine learning programs,'' in \emph{Proceedings of the 46th
  International Conference on Software Engineering: Software Engineering in
  Practice, {ICSE-SEIP} 2024, Lisbon, Portugal, April 14-20, 2024}.\hskip 1em
  plus 0.5em minus 0.4em\relax {ACM}, 2024, pp. 147--156. [Online]. Available:
  \url{https://doi.org/10.1145/3639477.3639718}
\BIBentrySTDinterwordspacing

\bibitem{DBLP:journals/pacmpl/DAntoniDGRRS24}
\BIBentryALTinterwordspacing
L.~D'Antoni, S.~Ding, A.~Goel, M.~Ramesh, N.~Rungta, and C.~Sung,
  ``Automatically reducing privilege for access control policies,'' \emph{Proc.
  {ACM} Program. Lang.}, vol.~8, no. {OOPSLA2}, pp. 763--790, 2024. [Online].
  Available: \url{https://doi.org/10.1145/3689738}
\BIBentrySTDinterwordspacing

\bibitem{DBLP:conf/icse/NongFYZL0C24}
\BIBentryALTinterwordspacing
Y.~Nong, R.~Fang, G.~Yi, K.~Zhao, X.~Luo, F.~Chen, and H.~Cai, ``{VGX:}
  large-scale sample generation for boosting learning-based software
  vulnerability analyses,'' in \emph{Proceedings of the 46th {IEEE/ACM}
  International Conference on Software Engineering, {ICSE} 2024, Lisbon,
  Portugal, April 14-20, 2024}.\hskip 1em plus 0.5em minus 0.4em\relax {ACM},
  2024, pp. 149:1--149:13. [Online]. Available:
  \url{https://doi.org/10.1145/3597503.3639116}
\BIBentrySTDinterwordspacing

\bibitem{DBLP:journals/nature/RulePCEKT24}
\BIBentryALTinterwordspacing
J.~S. Rule, S.~T. Piantadosi, A.~Cropper, K.~Ellis, M.~Nye, and J.~B.
  Tenenbaum, ``Symbolic metaprogram search improves learning efficiency and
  explains rule learning in humans,'' \emph{Nature Communications}, vol.~15,
  no.~1, p. 6847, 2024. [Online]. Available:
  \url{https://doi.org/10.1038/s41467-024-50966-x}
\BIBentrySTDinterwordspacing

\bibitem{DBLP:journals/pacmpl/BaderSPC19}
\BIBentryALTinterwordspacing
J.~Bader, A.~Scott, M.~Pradel, and S.~Chandra, ``Getafix: Learning to fix bugs
  automatically,'' \emph{Proc. ACM Program. Lang.}, vol.~3, no. OOPSLA, pp.
  159:1--159:27, 2019. [Online]. Available:
  \url{https://doi.org/10.1145/3360585}
\BIBentrySTDinterwordspacing

\bibitem{DBLP:conf/nsdi/MehtaBKMBABABK20}
\BIBentryALTinterwordspacing
S.~Mehta, R.~Bhagwan, R.~Kumar, C.~Bansal, C.~Maddila, B.~Ashok, S.~Asthana,
  C.~Bird, and A.~Kumar, ``Rex: Preventing bugs and misconfiguration in large
  services using correlated change analysis,'' in \emph{17th USENIX Symposium
  on Networked Systems Design and Implementation (NSDI 20)}.\hskip 1em plus
  0.5em minus 0.4em\relax Santa Clara, CA: USENIX Association, Feb. 2020, pp.
  435--448. [Online]. Available:
  \url{https://www.usenix.org/conference/nsdi20/presentation/mehta}
\BIBentrySTDinterwordspacing

\bibitem{DBLP:conf/sigsoft/WinterNBCHHWKWM22}
\BIBentryALTinterwordspacing
E.~R. Winter, V.~Nowack, D.~Bowes, S.~Counsell, T.~Hall, S.~{\'{O}}.
  Haraldsson, J.~R. Woodward, S.~Kirbas, E.~Windels, O.~McBello,
  A.~Atakishiyev, K.~Kells, and M.~W. Pagano, ``Towards developer-centered
  automatic program repair: findings from {Bloomberg},'' in \emph{Proceedings
  of the 30th {ACM} Joint European Software Engineering Conference and
  Symposium on the Foundations of Software Engineering, {ESEC/FSE} 2022,
  Singapore, Singapore, November 14-18, 2022}, A.~Roychoudhury, C.~Cadar, and
  M.~Kim, Eds.\hskip 1em plus 0.5em minus 0.4em\relax {ACM}, 2022, pp.
  1578--1588. [Online]. Available:
  \url{https://doi.org/10.1145/3540250.3558953}
\BIBentrySTDinterwordspacing

\bibitem{DBLP:conf/sbes/SousaSGBD21}
\BIBentryALTinterwordspacing
R.~{Rolim de Sousa}, G.~Soares, R.~Gheyi, T.~Barik, and L.~D'Antoni, ``Learning
  quick fixes from code repositories,'' in \emph{35th Brazilian Symposium on
  Software Engineering, {SBES} 2021, Joinville, Santa Catarina, Brazil, 27
  September 2021 - 1 October 2021}, C.~D. Vasconcellos, K.~G. Roggia,
  V.~Collere, and P.~Bousfield, Eds.\hskip 1em plus 0.5em minus 0.4em\relax
  {ACM}, 2021, pp. 74--83. [Online]. Available:
  \url{https://doi.org/10.1145/3474624.3474650}
\BIBentrySTDinterwordspacing

\bibitem{DBLP:conf/iwsc/YernauxV22}
\BIBentryALTinterwordspacing
G.~Yernaux and W.~Vanhoof, ``On detecting semantic clones in constraint logic
  programs,'' in \emph{16th {IEEE} International Workshop on Software Clones,
  {IWSC} 2022, Limassol, Cyprus, October 2, 2022}.\hskip 1em plus 0.5em minus
  0.4em\relax {IEEE}, 2022, pp. 32--38. [Online]. Available:
  \url{https://doi.org/10.1109/IWSC55060.2022.00014}
\BIBentrySTDinterwordspacing

\bibitem{DBLP:conf/ershov/BulychevKZ09}
\BIBentryALTinterwordspacing
P.~E. Bulychev, E.~V. Kostylev, and V.~A. Zakharov, ``Anti-unification
  algorithms and their applications in program analysis,'' in
  \emph{Perspectives of Systems Informatics, 7th International Andrei Ershov
  Memorial Conference, {PSI} 2009, Novosibirsk, Russia, June 15-19, 2009.
  Revised Papers}, ser. Lecture Notes in Computer Science, A.~Pnueli, I.~B.
  Virbitskaite, and A.~Voronkov, Eds., vol. 5947.\hskip 1em plus 0.5em minus
  0.4em\relax Springer, 2009, pp. 413--423. [Online]. Available:
  \url{https://doi.org/10.1007/978-3-642-11486-1\_35}
\BIBentrySTDinterwordspacing

\bibitem{DBLP:journals/programming/ThompsonLS17}
\BIBentryALTinterwordspacing
S.~J. Thompson, H.~Li, and A.~Schumacher, ``The pragmatics of clone detection
  and elimination,'' \emph{Art Sci. Eng. Program.}, vol.~1, no.~2, p.~8, 2017.
  [Online]. Available:
  \url{https://doi.org/10.22152/programming-journal.org/2017/1/8}
\BIBentrySTDinterwordspacing

\bibitem{gabbay}
\BIBentryALTinterwordspacing
M.~J. Gabbay, ``A theory of inductive definitions with alpha-equivalence,''
  Ph.D. dissertation, University of Cambridge, UK, 2000. [Online]. Available:
  \url{https://gabbay.org.uk/papers/thesis.pdf}
\BIBentrySTDinterwordspacing

\bibitem{DBLP:conf/lics/GabbayP99}
\BIBentryALTinterwordspacing
M.~Gabbay and A.~M. Pitts, ``A new approach to abstract syntax involving
  binders,'' in \emph{14th Annual {IEEE} Symposium on Logic in Computer
  Science, Trento, Italy, July 2-5, 1999}.\hskip 1em plus 0.5em minus
  0.4em\relax {IEEE} Computer Society, 1999, pp. 214--224. [Online]. Available:
  \url{https://doi.org/10.1109/LICS.1999.782617}
\BIBentrySTDinterwordspacing

\bibitem{DBLP:conf/fscd/Schmidt-Schauss22}
\BIBentryALTinterwordspacing
M.~Schmidt{-}Schau{\ss} and D.~Nantes{-}Sobrinho, ``Nominal anti-unification
  with atom-variables,'' in \emph{7th International Conference on Formal
  Structures for Computation and Deduction, {FSCD} 2022, August 2-5, 2022,
  Haifa, Israel}, ser. LIPIcs, A.~P. Felty, Ed., vol. 228.\hskip 1em plus 0.5em
  minus 0.4em\relax Schloss Dagstuhl - Leibniz-Zentrum f{\"{u}}r Informatik,
  2022, pp. 7:1--7:22. [Online]. Available:
  \url{https://doi.org/10.4230/LIPIcs.FSCD.2022.7}
\BIBentrySTDinterwordspacing

\bibitem{DBLP:conf/unif/BaumgartnerN20}
\BIBentryALTinterwordspacing
A.~Baumgartner and D.~Nantes{-}Sobrinho, ``{A}, {C}, and {AC} nominal
  anti-unification,'' in \emph{Proceedings of the 34th International Workshop
  on Unification, {UNIF} 2020, Linz, Austria, June 29, 2020}, T.~Kutsia and
  A.~M. Marshall, Eds., 2020, pp. 5:1--5:6. [Online]. Available:
  \url{http://www3.risc.jku.at/publications/download/risc\_6129/proceedings-UNIF2020.pdf\#page=27}
\BIBentrySTDinterwordspacing

\bibitem{DBLP:conf/rta/BaumgartnerKLV15}
\BIBentryALTinterwordspacing
A.~Baumgartner, T.~Kutsia, J.~Levy, and M.~Villaret, ``Nominal
  anti-unification,'' in \emph{26th International Conference on Rewriting
  Techniques and Applications, {RTA} 2015, June 29 to July 1, Warsaw, Poland},
  ser. LIPIcs, M.~Fern{\'{a}}ndez, Ed., vol.~36.\hskip 1em plus 0.5em minus
  0.4em\relax Schloss Dagstuhl - Leibniz-Zentrum fuer Informatik, 2015, pp.
  57--73. [Online]. Available:
  \url{http://dx.doi.org/10.4230/LIPIcs.RTA.2015.57}
\BIBentrySTDinterwordspacing

\bibitem{DBLP:journals/jar/KutsiaLV14}
\BIBentryALTinterwordspacing
T.~Kutsia, J.~Levy, and M.~Villaret, ``Anti-unification for unranked terms and
  hedges,'' \emph{J. Autom. Reason.}, vol.~52, no.~2, pp. 155--190, 2014.
  [Online]. Available: \url{https://doi.org/10.1007/s10817-013-9285-6}
\BIBentrySTDinterwordspacing

\bibitem{DBLP:journals/iandc/BaumgartnerK17}
\BIBentryALTinterwordspacing
A.~Baumgartner and T.~Kutsia, ``Unranked second-order anti-unification,''
  \emph{Inf. Comput.}, vol. 255, pp. 262--286, 2017. [Online]. Available:
  \url{https://doi.org/10.1016/j.ic.2017.01.005}
\BIBentrySTDinterwordspacing

\bibitem{DBLP:conf/jelia/BaumgartnerK14}
\BIBentryALTinterwordspacing
------, ``A library of anti-unification algorithms,'' in \emph{Logics in
  Artificial Intelligence - 14th European Conference, {JELIA} 2014, Funchal,
  Madeira, Portugal, September 24-26, 2014. Proceedings}, ser. Lecture Notes in
  Computer Science, E.~Ferm{\'{e}} and J.~Leite, Eds., vol. 8761.\hskip 1em
  plus 0.5em minus 0.4em\relax Springer, 2014, pp. 543--557. [Online].
  Available: \url{https://doi.org/10.1007/978-3-319-11558-0\_38}
\BIBentrySTDinterwordspacing

\bibitem{RISC5180}
\BIBentryALTinterwordspacing
A.~Baumgartner, ``\BIBforeignlanguage{english}{{Anti-Unification Algorithms:
  Design, Analysis, and Implementation}},'' Ph.D. dissertation, RISC, JKU Linz,
  September 2015. [Online]. Available:
  \url{http://www.risc.jku.at/publications/download/risc_5180/phd-thesis.pdf}
\BIBentrySTDinterwordspacing

\bibitem{DBLP:conf/ilp/YamamotoIIA01}
\BIBentryALTinterwordspacing
A.~Yamamoto, K.~Ito, A.~Ishino, and H.~Arimura, ``Modelling semi-structured
  documents with hedges for deduction and induction,'' in \emph{Inductive Logic
  Programming, 11th International Conference, {ILP} 2001, Strasbourg, France,
  September 9-11, 2001, Proceedings}, ser. Lecture Notes in Computer Science,
  C.~Rouveirol and M.~Sebag, Eds., vol. 2157.\hskip 1em plus 0.5em minus
  0.4em\relax Springer, 2001, pp. 240--247. [Online]. Available:
  \url{https://doi.org/10.1007/3-540-44797-0\_20}
\BIBentrySTDinterwordspacing

\bibitem{DBLP:conf/tbillc/DunduaKR19}
\BIBentryALTinterwordspacing
B.~Dundua, T.~Kutsia, and M.~Rukhaia, ``Unranked nominal unification,'' in
  \emph{Language, Logic, and Computation - 13th International Tbilisi
  Symposium, TbiLLC 2019, Batumi, Georgia, September 16-20, 2019, Revised
  Selected Papers}, ser. Lecture Notes in Computer Science,
  A.~{\"{O}}zg{\"{u}}n and Y.~Zinova, Eds., vol. 13206.\hskip 1em plus 0.5em
  minus 0.4em\relax Springer, 2019, pp. 279--296. [Online]. Available:
  \url{https://doi.org/10.1007/978-3-030-98479-3\_14}
\BIBentrySTDinterwordspacing

\bibitem{DBLP:journals/scp/RoyCK09}
\BIBentryALTinterwordspacing
C.~K. Roy, J.~R. Cordy, and R.~Koschke, ``Comparison and evaluation of code
  clone detection techniques and tools: {A} qualitative approach,'' \emph{Sci.
  Comput. Program.}, vol.~74, no.~7, pp. 470--495, 2009. [Online]. Available:
  \url{https://doi.org/10.1016/j.scico.2009.02.007}
\BIBentrySTDinterwordspacing

\bibitem{DBLP:journals/spe/Yang91}
\BIBentryALTinterwordspacing
W.~Yang, ``Identifying syntactic differences between two programs,''
  \emph{Softw. Pract. Exp.}, vol.~21, no.~7, pp. 739--755, 1991. [Online].
  Available: \url{https://doi.org/10.1002/spe.4380210706}
\BIBentrySTDinterwordspacing

\bibitem{DBLP:journals/sqj/EvansFM09}
\BIBentryALTinterwordspacing
W.~S. Evans, C.~W. Fraser, and F.~Ma, ``Clone detection via structural
  abstraction,'' \emph{Softw. Qual. J.}, vol.~17, no.~4, pp. 309--330, 2009.
  [Online]. Available: \url{https://doi.org/10.1007/s11219-009-9074-y}
\BIBentrySTDinterwordspacing

\bibitem{DBLP:conf/icsm/BaxterYMSB98}
\BIBentryALTinterwordspacing
I.~D. Baxter, A.~Yahin, L.~M. de~Moura, M.~Sant'Anna, and L.~Bier, ``Clone
  detection using abstract syntax trees,'' in \emph{1998 International
  Conference on Software Maintenance, {ICSM} 1998, Bethesda, Maryland, USA,
  November 16-19, 1998}.\hskip 1em plus 0.5em minus 0.4em\relax {IEEE} Computer
  Society, 1998, pp. 368--377. [Online]. Available:
  \url{https://doi.org/10.1109/ICSM.1998.738528}
\BIBentrySTDinterwordspacing

\bibitem{DBLP:conf/wcre/KoschkeFF06}
\BIBentryALTinterwordspacing
R.~Koschke, R.~Falke, and P.~Frenzel, ``Clone detection using abstract syntax
  suffix trees,'' in \emph{13th Working Conference on Reverse Engineering
  {(WCRE} 2006), 23-27 October 2006, Benevento, Italy}.\hskip 1em plus 0.5em
  minus 0.4em\relax {IEEE} Computer Society, 2006, pp. 253--262. [Online].
  Available: \url{https://doi.org/10.1109/WCRE.2006.18}
\BIBentrySTDinterwordspacing

\bibitem{DBLP:conf/padl/LiT10}
\BIBentryALTinterwordspacing
H.~Li and S.~J. Thompson, ``Similar code detection and elimination for erlang
  programs,'' in \emph{Practical Aspects of Declarative Languages, 12th
  International Symposium, {PADL} 2010, Madrid, Spain, January 18-19, 2010.
  Proceedings}, ser. Lecture Notes in Computer Science, M.~Carro and
  R.~Pe{\~{n}}a, Eds., vol. 5937.\hskip 1em plus 0.5em minus 0.4em\relax
  Springer, 2010, pp. 104--118. [Online]. Available:
  \url{https://doi.org/10.1007/978-3-642-11503-5\_10}
\BIBentrySTDinterwordspacing

\bibitem{DBLP:journals/tocl/LevyV12}
\BIBentryALTinterwordspacing
J.~Levy and M.~Villaret, ``Nominal unification from a higher-order
  perspective,'' \emph{{ACM} Trans. Comput. Log.}, vol.~13, no.~2, pp.
  10:1--10:31, 2012. [Online]. Available:
  \url{https://doi.org/10.1145/2159531.2159532}
\BIBentrySTDinterwordspacing

\end{thebibliography}
\end{document}